\pdfoutput=1
\documentclass[lettersize,journal]{IEEEtran}

\usepackage{amsmath,amsfonts,amssymb}
\usepackage{algorithmic}
\usepackage{algorithm}
\usepackage{array}
\usepackage{multirow}
\usepackage[caption=false,font=normalsize,labelfont=sf,textfont=sf]{subfig}
\usepackage{textcomp}
\usepackage{stfloats}
\usepackage{url}
\usepackage{verbatim}
\usepackage{graphicx}
\usepackage{cite}
\usepackage{booktabs}
\usepackage{tikz}
\usetikzlibrary{arrows.meta,calc,fit,positioning}
\newtheorem{assumption}{Assumption}
\newtheorem{definition}{Definition}
\newtheorem{lemma}{Lemma}
\newtheorem{theorem}{Theorem}
\newtheorem{proposition}{Proposition}
\newtheorem{corollary}{Corollary}
\newtheorem{remark}{Remark}

\graphicspath{{figures/}}

\begin{document}

\title{CC-OPI: Online Distributed Task Allocation for UAV Swarms under Communication Constraints}

\author{Liu Biao and Zhang Tong}

\maketitle

\begin{abstract}
In multi-robot missions such as post-disaster search and rescue, a short communication range fragments a swarm of Unmanned Aerial Vehicles (UAVs) into transient information islands. Under such intermittent connectivity, the prevailing ``allocate-then-execute'' paradigm---which requires global consensus before any physical movement---breaks down. This paper proposes the Communication-Constrained Online Performance Impact (CC-OPI) algorithm, an event-driven method that interleaves task negotiation with physical execution. CC-OPI replans only at discrete physical and topological events and integrates two further elements. The first is a pair of cost-evaluation metrics adapted to dynamic topologies---one with a spatial locality penalty that promotes regionalized operation, the other with a deadline-aware urgency term---complemented by a non-preemptive state lock that shields each UAV's ongoing action. The second is a decentralized fault-tolerance layer that pairs version-based state synchronization with a global-time-driven emergency pool. We establish that CC-OPI terminates in finite time, free of stale-completion deadlock and of unbounded reassignment within the mission horizon. In simulations at a $250$~m communication radius, CC-OPI sustains a task completion rate of about $0.80$: it leads a matched online execution of the unmodified Performance Impact (PI) and Consensus-Based Bundle Algorithm (CBBA) rules by about seven percentage points, exceeds the naively transferred static baselines by roughly $20$ points, and remains within several points of PI and CBBA under full connectivity. Within the tested settings, CC-OPI degrades gracefully as connectivity weakens and absorbs packet loss, terrain occlusion, and runtime task arrival. The price is more messages and some redundant travel---a deliberate trade-off of efficiency for robustness.
\end{abstract}


\section{Introduction}\label{sec:introduction}
\IEEEPARstart{M}{ulti-Agent} Systems (MASs) are increasingly central to complex, time-sensitive missions such as post-disaster search and rescue \cite{queralta2020collaborative, erdelj2017help} or automated warehouse logistics \cite{wurman2008coordinating}. At the core of such collaborative operations lies the Single-Task, Single-Robot, Time-Extended Assignment (ST-SR-TA) problem \cite{gerkey2004formal, korsah2013comprehensive}, in which agents autonomously construct sequential execution paths to serve spatially distributed tasks under strict deadline constraints. In real-world deployments, meeting these time windows while sustaining high overall execution efficiency is often decisive for mission success.

To address the inherent computational complexity of task allocation \cite{gerkey2004formal}, distributed heuristic algorithms---most notably the Consensus-Based Bundle Algorithm (CBBA) \cite{choi2009consensus} and the Performance Impact (PI) algorithm \cite{zhao2015heuristic}---have been studied extensively. These methods attain near-optimal solutions through decentralized coordination among autonomous agents. However, their convergence guarantees rest on the network staying connected during negotiation---or, in the weaker dynamic-topology variants, on the union of the communication graphs over a bounded window remaining connected, so that every message eventually propagates network-wide. Accordingly, these methods typically follow a static, offline ``allocate-then-execute'' paradigm, in which agents reach a global consensus through multiple rounds of communication before any physical movement begins \cite{otte2020auctions}.

In post-disaster deployments, damaged infrastructure and limited onboard hardware often restrict the communication range of Unmanned Aerial Vehicles (UAVs) to a short distance (e.g., 250~m) \cite{hayat2016survey, asadpour2013characterizing, mozaffari2019tutorial, zeng2016wireless}. As a result, the global communication network is frequently fragmented into multiple dynamic ``information islands'' \cite{vahdat2000epidemic}. This spatial isolation creates a fundamental paradox: information can be exchanged only after UAVs physically move to reconfigure the network topology, yet such connectivity-seeking movement competes with the timely execution of tasks. The resulting execution dilemma is twofold. Prioritizing task execution alone leads to prolonged disconnection, leaving local state views outdated and causing tasks to be omitted; conversely, deviating from efficient trajectories to preserve connectivity incurs deadline violations. Consequently, applying the traditional static planning paradigm in such intermittently connected networks induces blind cross-regional movement and persistent assignment oscillation, degrading overall system performance. Fig.~\ref{fig:mot_islands} illustrates this predicament.

\begin{figure*}[!t]
\centering
\subfloat[Information islands under a short communication radius $R_c$]{%
    \includegraphics[width=0.9\columnwidth]{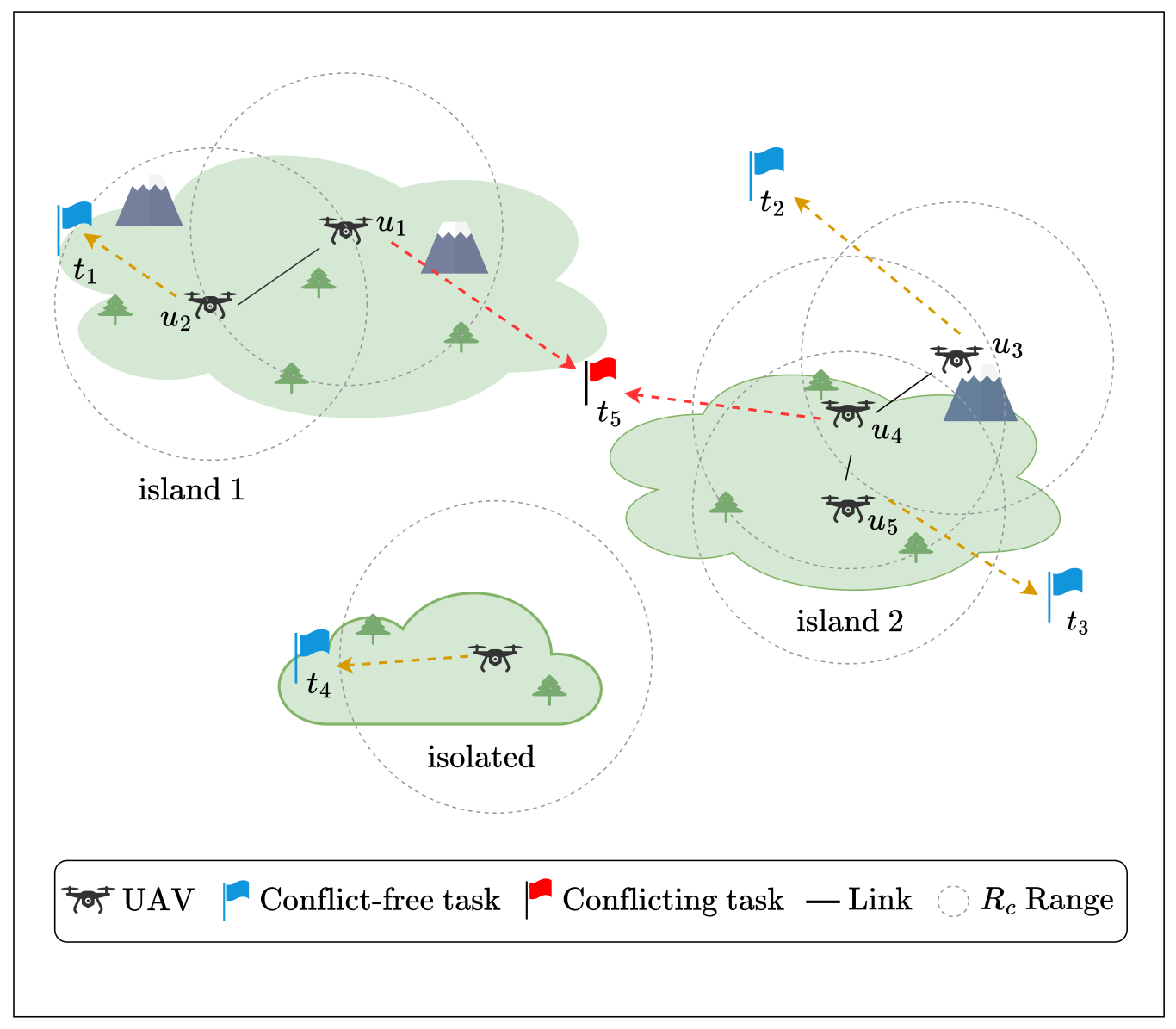}%
    \label{fig:mot_islands}}
\hfil
\subfloat[Static ``allocate-then-execute'' versus CC-OPI's online ``allocate-while-executing'' paradigm]{%
    \includegraphics[width=0.9\columnwidth]{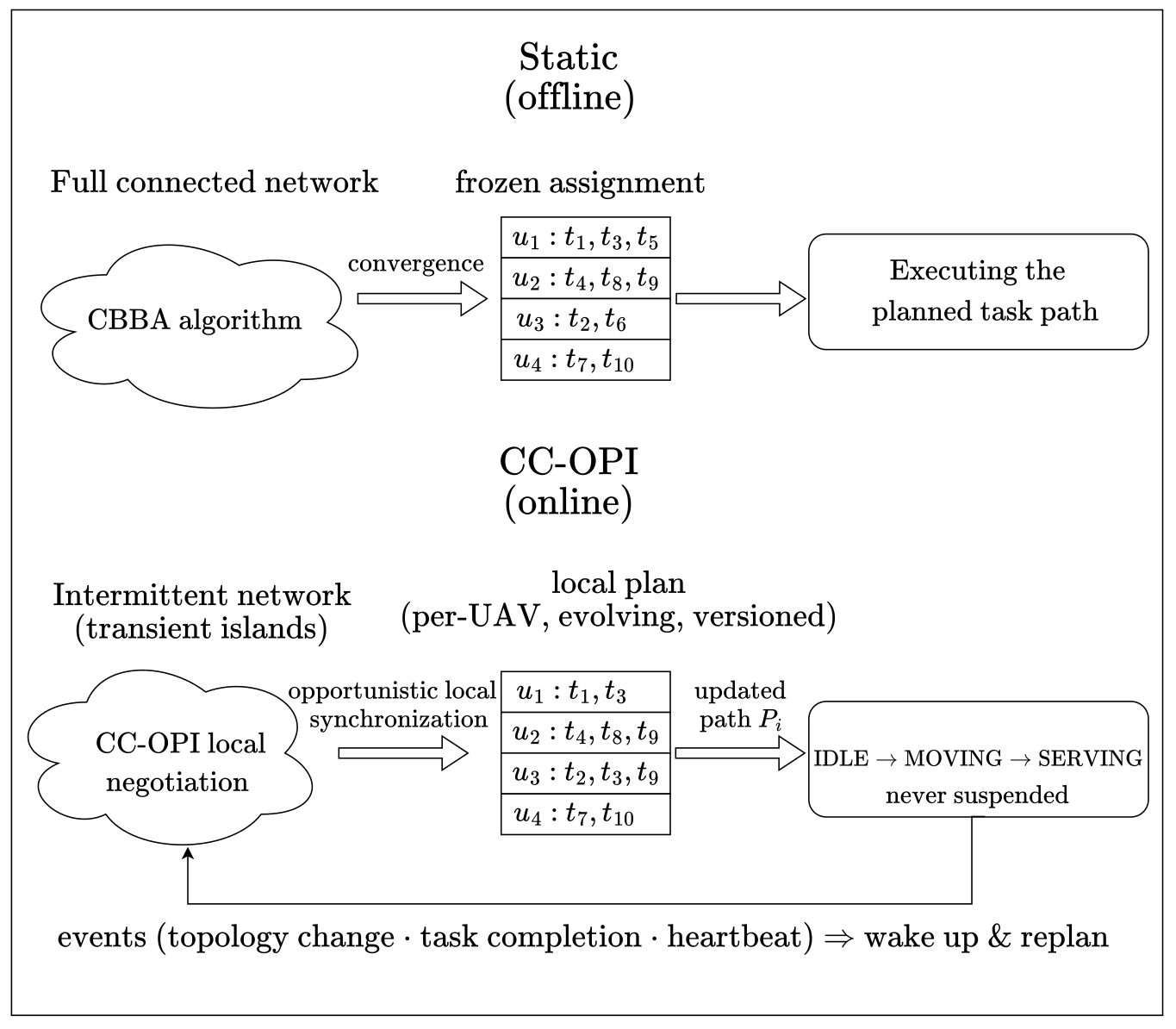}%
    \label{fig:mot_paradigm}}
\caption{Motivating scenario and paradigm comparison.}
\label{fig:motivation}
\end{figure*}

To overcome these limitations, this paper proposes the Communication-Constrained Online Performance Impact (CC-OPI) algorithm, which introduces an event-driven ``allocate-while-executing'' online paradigm (Fig.~\ref{fig:mot_paradigm}). By decoupling algorithmic consensus from continuous connectivity, CC-OPI allows UAVs to sustain implicit coordination---conflict-reducing cooperation that emerges from a spatial partitioning induced by a spatial locality penalty, without explicit message exchange---even during communication disruptions. Specifically, it redefines the Insertion Performance Impact (IPI) with this penalty, which encourages the UAV swarm to adopt a regionalized operating pattern. In addition, it reformulates the Removal Performance Impact (RPI) to prioritize near-deadline tasks, and pairs it with a non-preemptive state-locking mechanism that shields each UAV's ongoing physical action, thereby suppressing assignment oscillation when partitioned subnets merge.

The main contributions of this paper are summarized as follows:
\begin{enumerate}
    \item \textbf{An Online Allocation Framework for Persistent Partitioning:} We carry event-driven replanning during execution---itself an established idea in asynchronous distributed allocation (Section~\ref{sec:related})---into the regime in which the network may remain partitioned throughout the mission, with no joint-connectivity guarantee. Rather than separating planning from execution as in static paradigms, the framework allows UAVs' physical movement and task negotiation to proceed asynchronously and concurrently.
    \item \textbf{Cost-Evaluation Mechanisms Adapted to Dynamic Topologies:} We design marginal performance-evaluation functions for dynamic network topologies. These functions center on a spatial locality penalty, retain deadline-aware urgency weighting as part of the overall scoring design, and are complemented by a non-preemptive state-locking mechanism that protects ongoing physical actions, mitigating blind cross-regional task chasing and assignment oscillation under intermittent connectivity.
    \item \textbf{Decentralized State Synchronization and Fault Tolerance:} We develop a decentralized state-synchronization and fault-tolerance mechanism. Building on a monotonic, version-based consistency protocol and a global-time-driven Emergency Pool, it averts the stale-information coordination deadlocks and task omissions that arise in fragmented networks.
\end{enumerate}

The remainder of this paper is organized as follows. Section~\ref{sec:related} reviews related work on distributed task allocation and multi-agent coordination under communication constraints. Section~\ref{problem} presents the system model and formally defines the optimization problem. Section~\ref{sec:cc-opi} details the core mechanisms and design of the proposed CC-OPI algorithm. Section~\ref{sec:theoretical} establishes the finite-time termination of the algorithm and its freedom from stale-completion deadlock and unbounded reassignment, together with its computational complexity. Section~\ref{sec:simulation} evaluates the algorithm through extensive simulation experiments. Section~\ref{sec:conclusion} concludes the paper and outlines future research directions.

\section{Related Work}
\label{sec:related}

\subsection{Distributed Task Allocation}
Early work on multi-agent task allocation relied largely on centralized optimization \cite{kuhn1955hungarian}, which achieves global optimality in principle but scales poorly and is vulnerable to single points of failure; the field has therefore shifted toward distributed coordination, in which market-based auction mechanisms are a mainstream solution \cite{smith1980contract, dias2006market}. Among distributed auction approaches, CBBA \cite{choi2009consensus} and the PI algorithm \cite{zhao2015heuristic} are the classical baselines. CBBA couples distributed auctions with a consensus protocol: agents iteratively build task bundles from marginal scores and resolve conflicts through predefined consensus rules. The PI algorithm instead evaluates the marginal impact of a candidate task on an agent's current local path, alternating inclusion and removal phases to reach a conflict-free assignment. A later extension, referred to as PI-MaxAss \cite{turner2017distributed}, augments PI with a two-stage rescheduling procedure that explicitly maximizes the number of assigned tasks, improving task coverage when the network is fully connected; more recent distributed methods further pursue this maximum-assignment objective \cite{wang2024efficient}. Task allocation under temporal and ordering constraints has been systematically categorized \cite{nunes2017taxonomy}, and deadline-constrained distributed assignment has been studied explicitly \cite{luo2015distributed}; across this line of work, auction-based methods converge to near-optimal solutions using only local, iterative communication \cite{ka2024systematic}.

These algorithms do not, strictly speaking, require a complete communication graph: PI was evaluated on sparse local topologies such as row, star, and mesh networks \cite{zhao2015heuristic}, and the convergence theory of CBBA admits dynamic networks, provided the union of the communication graphs over a bounded window remains connected \cite{choi2009consensus}. What they share is the weaker but still demanding premise that sufficient information propagation is achieved---through joint connectivity or eventually reliable message delivery---within the negotiation horizon. Both analytical and experimental studies show that once communication falls below this level, the solution quality of such static auction and consensus methods degrades substantially \cite{otte2020auctions, nayak2020experimental}. The ``allocate-then-execute'' paradigm, moreover, requires agents to reach a global assignment consensus before any physical movement, and is fragile under strict communication-range limits and frequent network partitioning: agents isolated in separate information islands may, from outdated local views, admit the same task into multiple execution paths or fail to respond to time-critical tasks altogether.

\subsection{Task Allocation under Communication Constraints}
Research on task allocation under restricted communication in MASs follows two main directions. The first seeks to maintain continuous network connectivity throughout the mission \cite{zavlanos2011graph, ponda2012distributed}, typically by deploying dedicated communication relay nodes or imposing line-of-sight and distance constraints on agent trajectories. In environments with severely limited communication radii (e.g., $R_c \le 250$~m) and complex terrain, however, enforcing continuous connectivity is costly: it restricts the physical coverage of the UAV swarm and effectively reduces high-performance rescue UAVs to mobile relays, undermining the primary mission objectives.

The second direction---which this paper also adopts---relaxes continuous connectivity and instead exploits local or opportunistic communication, for example through optimized local task swaps \cite{liu2015communication} or the Delay-Tolerant Network paradigm \cite{fall2003delay}. In this setting, agents act on local, possibly stale information and synchronize state only during opportunistic physical encounters. Asynchronous consensus auctions, notably ACBBA \cite{johnson2011asynchronous}, likewise remove the synchronized-round structure and allow agents to negotiate over locally timed message exchanges; they still presume, however, that messages are eventually delivered across the network, and they do not account for the physical commitment a UAV accumulates while executing. These relaxed-communication strategies share a further fundamental limitation---their marginal performance-evaluation metrics remain ``myopic'' with respect to the disconnected state of the network, implicitly treating an isolated view as the global ground truth. This lack of partition awareness has two consequences. First, agents in different information islands tend to lock onto the same remote, high-value task, causing spatial conflicts and redundant deployment. Second, when disconnected subnets merge, consensus resolves conflicts using only the current numerical bids, disregarding the movement cost already invested; the resulting preemption induces persistent assignment oscillation and often causes preempted agents to miss the deadlines of alternative tasks.

Recent work extends these directions along four axes. Under intermittent or unreliable links, communication-aware consensus auctions weigh the value of each exchange against its cost before transmitting \cite{raja2022communication}, and dynamic allocation protocols have been hardened against stale, asymmetric views of the assignment \cite{zhao2026assignment}. Event-triggered coordination places replanning-on-significant-change on a principled footing \cite{nowzari2019event}, and bundle-auction extensions replan locally when time-sensitive tasks change during a mission \cite{chen2022local}. A third axis co-designs motion and communication, either by steering robots to restore connectivity intermittently under temporal-logic task specifications \cite{kantaros2019temporal} or through allocation rules that themselves induce connectivity in a flying ad hoc network \cite{leong2024scalable}. A fourth axis strengthens systems validation, coupling robot and network simulators \cite{acharya2020cornet} and carrying swarm autonomy into urban field experimentation \cite{chung2023darpa}. Across these axes, however, reconciliation is still assumed to complete within bounded contact windows, or connectivity is treated as an objective the robots must actively restore; in either case the allocation scores remain decoupled from the physical commitment accumulated during execution---the precise combination that the persistently partitioned regime demands.

In contrast, the proposed CC-OPI framework targets the regime in which no joint-connectivity guarantee holds: the network may remain partitioned for a large fraction of the mission, the topology changes are induced by the UAVs' own task-driven movement, and state reconciliation is therefore opportunistic and must proceed while physical execution continues. To operate in this regime, CC-OPI gives UAVs an implicit-coordination capability during communication disruptions. A dynamic spatial locality penalty factor $\mu$ guides the swarm toward regionalized operation, which discourages blind cross-regional assignment and reduces the overlapping claims that drive assignment oscillation when fragmented subnets merge. A non-preemptive state-locking mechanism complements this spatial regularization by shielding a UAV's ongoing service and its final, irreversible approach during opportunistic consensus. We note that spatial regularization, non-preemptive locking, monotone gossip, and deadline-triggered prioritization each have precedents in neighboring literatures; the contribution claimed here is their disciplined integration into a single online allocation protocol, together with the specific formulations---the fixed deployment anchor and the conditional lifting of the spatial penalty for unattended emergencies---that make the combination effective under persistent partitioning. A capability mapping of the cited methods across six dimensions---connectivity premise, planning--execution coupling, commitment protection, completion-state dissemination, runtime task arrival, and communication pattern---is tabulated in the supplementary material; a static geographic-partition baseline representing the territory-based alternative is evaluated experimentally in Section~\ref{sec:simulation}.

\subsection{Online Planning and Dynamic Adaptation}
Practical deployments, particularly search and rescue, must also cope with the unpredictable emergence of new tasks (e.g., newly discovered trapped victims). A common approach in traditional distributed allocation handles such dynamics through periodic or event-triggered global replanning under a ``time-freezing'' abstraction, which suspends the physical movement of the entire UAV swarm until a new round of consensus has converged \cite{chen2019distributed}. This again presumes a persistent, fully connected network: under partitioning, the global ``pause'' or ``replan'' signal cannot reach isolated subnets, so newly emerged time-critical tasks are overlooked or assigned in conflict.

CC-OPI instead dispenses with the staged, static planning paradigm and adopts an event-driven, decentralized execution model \cite{heemels2012introduction, johnson2011asynchronous}: a task that emerges mid-mission is treated as an ordinary local trigger, processed with the same priority as a topology change or a task completion, without requiring network-wide suspension. We therefore use dynamically arriving tasks as an additional scenario to assess robustness rather than as a core capability.

\section{Problem Formulation}
\label{problem}

\subsection{System Model}
Consider a post-disaster search and rescue scenario in which a swarm of heterogeneous UAVs is deployed to deliver critical supplies (e.g., medical kits or food) to spatially distributed trapped individuals. The set of UAVs is denoted by $\mathcal{U} = \{u_1, u_2, \dots, u_N\}$, where each UAV $u_i \in \mathcal{U}$ is described by the tuple:
\begin{equation}
u_i \triangleq \langle \text{type}_i, v_i, C_i, p_i(t) \rangle.
\end{equation}
Here, $\text{type}_i \in \{\text{medicine}, \text{food}\}$ is the supply type that $u_i$ can transport; $v_i$ is its constant cruising speed; $C_i$ is the maximum payload capacity, defined as the upper bound on the number of tasks in its execution path; and $p_i(t) \in \mathbb{R}^2$ is its position at time $t$.

The set of search and rescue tasks is denoted by $\mathcal{T} = \{\tau_1, \tau_2, \dots, \tau_M\}$, where each task $\tau_k \in \mathcal{T}$ is described by the tuple:
\begin{equation}
\tau_k \triangleq \langle \text{type}_k, p_k, s_k, D_k \rangle.
\end{equation}
Here, $\text{type}_k \in \{\text{medicine}, \text{food}\}$ is the required supply type; $p_k \in \mathbb{R}^2$ is the task location; $s_k$ is the deterministic service time required upon arrival; and $D_k$ is the hard deadline. A task $\tau_k$ is completed if and only if a type-compatible UAV $u_i$ (i.e., $\text{type}_i = \text{type}_k$) arrives at $p_k$ at a time $A_{i,k}$ satisfying $A_{i,k} \le D_k$, with equality counting as on time; service may finish after the deadline. Each task is served by at most one UAV: a UAV that arrives at a task already completed, or under service by a peer, abandons it, so no task is served concurrently or repeatedly.

\subsection{Distance-Dependent Communication Network}
\label{subsec:comm_network}
In post-disaster environments, establishing and maintaining a persistent, fully connected communication backbone is generally impractical because of the destruction of ground infrastructure and the limited transmission power of mobile UAVs. The UAV swarm must therefore rely on opportunistic short-range peer-to-peer communication links. We model this dynamic topology as a time-varying undirected graph $\mathcal{G}(t) = (\mathcal{U}, \mathcal{E}(t))$, where $\mathcal{U}$ is the set of UAVs and $\mathcal{E}(t)$ is the set of active communication edges at time $t$.

A communication link between two UAVs is determined by their spatial proximity. Specifically, a communication edge $e_{ij}(t) \in \mathcal{E}(t)$ exists between UAVs $u_i$ and $u_j$ if and only if their Euclidean distance does not exceed the maximum communication radius $R_c$:
\begin{equation}
\mathcal{E}(t) = \{e_{ij}(t) \mid \|p_i(t) - p_j(t)\| \le R_c, \forall u_i, u_j \in \mathcal{U}, i \neq j\},
\end{equation}
where $R_c$ is a hardware-specific constant representing the physical limit of reliable data transmission.

Because the operational workspace is large relative to the communication radius $R_c$, the graph $\mathcal{G}(t)$ is typically sparse and frequently disconnected throughout task execution, partitioning the UAV swarm into disjoint information islands. Global synchronization is therefore infeasible \cite{fall2003delay}, and each UAV $u_i$ must maintain an independent local state view $S_i(t)$, which is subject to information lag and may deviate from the true global state, especially regarding UAVs and tasks in remote subgraphs.

\subsection{Task Execution Kinematics}
Each UAV $u_i$ maintains an ordered execution path $P_i = [\tau_{i,1}, \tau_{i,2}, \dots, \tau_{i,|P_i|}] \subseteq \mathcal{T}$, subject to the payload capacity constraint $|P_i| \le C_i$. This ordered structure couples the UAV's spatial trajectory with its temporal schedule.

Execution timing is governed by the estimated arrival times of the tasks along the path. Evaluated at the current physical time $t$, the estimated arrival time $A_{i,k}$ of UAV $u_i$ at the $k$-th task $\tau_{i,k}$ in its path $P_i$ is computed recursively. For the first task in the path ($k=1$), the arrival time depends on the current position $p_i(t)$ of the UAV and its constant cruising speed $v_i$:
\begin{equation}
A_{i,1} = t + \frac{\|p_i(t) - p_{\tau_{i,1}}\|}{v_i},
\end{equation}
where $p_{\tau_{i,1}}$ denotes the position of task $\tau_{i,1}$. For any subsequent task ($k > 1$), the arrival time is the sum of the arrival time at the preceding task, its deterministic service time $s_{\tau_{i,k-1}}$, and the flight time between the two successive task locations:
\begin{equation}
A_{i,k} = A_{i,k-1} + s_{\tau_{i,k-1}} + \frac{\|p_{\tau_{i,k-1}} - p_{\tau_{i,k}}\|}{v_i}.
\end{equation}

Because rescue operations are time-critical, deadline satisfaction is mandatory. A path $P_i$ is feasible if and only if the UAV arrives no later than the hard deadline of every assigned task:
\begin{equation}
A_{i,k} \le D_{\tau_{i,k}}, \quad \forall \tau_{i,k} \in P_i.
\end{equation}
Any path that violates this condition is treated as infeasible during allocation.

\subsection{Optimization Objective and Challenge}
Assuming a centralized controller with global knowledge of all UAVs and tasks, the task allocation problem can be cast as a lexicographic multi-objective optimization. The primary objective maximizes the number of assigned tasks \cite{turner2017distributed}, while the secondary objective improves overall operational efficiency. Let $x_{i,k} \in \{0, 1\}$ be a binary decision variable, where $x_{i,k} = 1$ if task $\tau_k$ is assigned to UAV $u_i$, and $x_{i,k} = 0$ otherwise. The global optimization objective is defined as:
\begin{equation}
    \text{lex} \min_{\mathbf{X}} \left[ -\mathcal{F}_1(\mathbf{X}), \mathcal{F}_2(\mathbf{X}) \right]
\end{equation}
subject to:
\begin{align}
    \sum_{i=1}^N x_{i,k} \le 1, & \quad \forall \tau_k \in \mathcal{T}, \\
    |P_i| \le C_i, & \quad \forall u_i \in \mathcal{U}, \\
    x_{i,k} = 1 \Rightarrow \text{type}_i = \text{type}_k, & \quad \forall u_i \in \mathcal{U}, \tau_k \in \mathcal{T}, \\
    x_{i,k} = 1 \Rightarrow A_{i,k} \le D_k, & \quad \forall u_i \in \mathcal{U}, \tau_k \in \mathcal{T},
\end{align}
where $\mathcal{F}_1 = \sum_{i=1}^N \sum_{k=1}^M x_{i,k}$ denotes the total number of assigned tasks---each guaranteed to be completed on time by the deadline constraint---and $\mathcal{F}_2 = \sum_{i=1}^N \sum_{\tau_k \in P_i} A_{i,k}$ denotes the MinSum objective \cite{zhao2015heuristic}, i.e., the aggregate arrival time over all assigned tasks.

However, under the communication limitations described in Section~\ref{subsec:comm_network}, this centralized optimum is unattainable in practice. Traditional distributed heuristics (e.g., CBBA \cite{choi2009consensus} and PI \cite{zhao2015heuristic}) require the topology graph $\mathcal{G}(t)$ to remain connected, or at least sufficiently periodically connected, to guarantee global state convergence; under the fragmentation described above, static consensus cannot achieve network-wide synchronization, and the inconsistent local views $S_i(t)$ lead to conflicting assignments or omitted tasks.

This dilemma is the primary motivation for the CC-OPI algorithm: rather than pursuing a static global optimum, we develop a fully decentralized, event-driven online mechanism. By allowing UAVs to evaluate and synchronize their assignment states concurrently with physical execution, CC-OPI seeks to approximate the global optimum despite persistent and unpredictable communication disruptions.

Two clarifications connect this formulation to the evaluation in Section~\ref{sec:simulation}. First, the lexicographic objective is a centralized \emph{reference} objective rather than a quantity any UAV optimizes directly: CC-OPI optimizes local marginal surrogates of $\mathcal{F}_1$ and $\mathcal{F}_2$ (Section~\ref{sec:cc-opi}), and under stale views two partitions may transiently assign the same task, so the number of locally assigned feasible tasks need not equal the number of tasks uniquely completed. The empirical counterpart of $\mathcal{F}_1$ is therefore the realized task completion rate---the fraction of tasks reached by their deadline and served exactly once, uniqueness following from the single-service rule above and completion knowledge spreading through the synchronization layer of Section~\ref{sec:cc-opi}. Second, a task that no type-compatible UAV could reach in time even by direct flight admits no feasible assignment under the deadline constraint; such tasks are excluded from the completion-rate denominator, and dynamically arriving tasks, evaluated from the swarm state at their arrival, are treated by the same rule (Section~\ref{sec:simulation}).

\section{The CC-OPI Algorithm}
\label{sec:cc-opi}

To address the communication limitations and pursue the optimization objectives established in Section~\ref{problem}, this section introduces the CC-OPI algorithm. CC-OPI restructures the task allocation process, replacing the traditional static paradigm with an event-driven online execution model in which task evaluation and state synchronization are triggered by discrete physical and topological changes rather than by synchronous iteration. By reformulating the marginal performance-evaluation metrics and adding a decentralized fault-tolerance mechanism, CC-OPI allows UAVs to approximate the global optimization objectives through local interaction and implicit coordination, even under persistent communication disruptions.

\subsection{Event-Driven Online Execution Paradigm}
Traditional distributed task allocation relies on network-wide information propagation to reach the global consensus that drives each algorithm iteration. In networks fragmented by communication constraints, this reliance leaves each UAV with local or outdated information, producing cascading iteration errors and, ultimately, allocation failures. To avoid this, CC-OPI adopts an asynchronous, event-driven architecture that decouples the physical dynamics of the UAVs from the negotiation process. Under this paradigm, physical execution (e.g., flying toward a target or performing a rescue) is never suspended to wait for network synchronization. Instead, each UAV continuously follows its locally maintained execution path and replans only when a specific local event occurs. Fig.~\ref{fig:arch} depicts this architecture at a single UAV.

\begin{figure}[!t]
\centering
\includegraphics[width=\columnwidth]{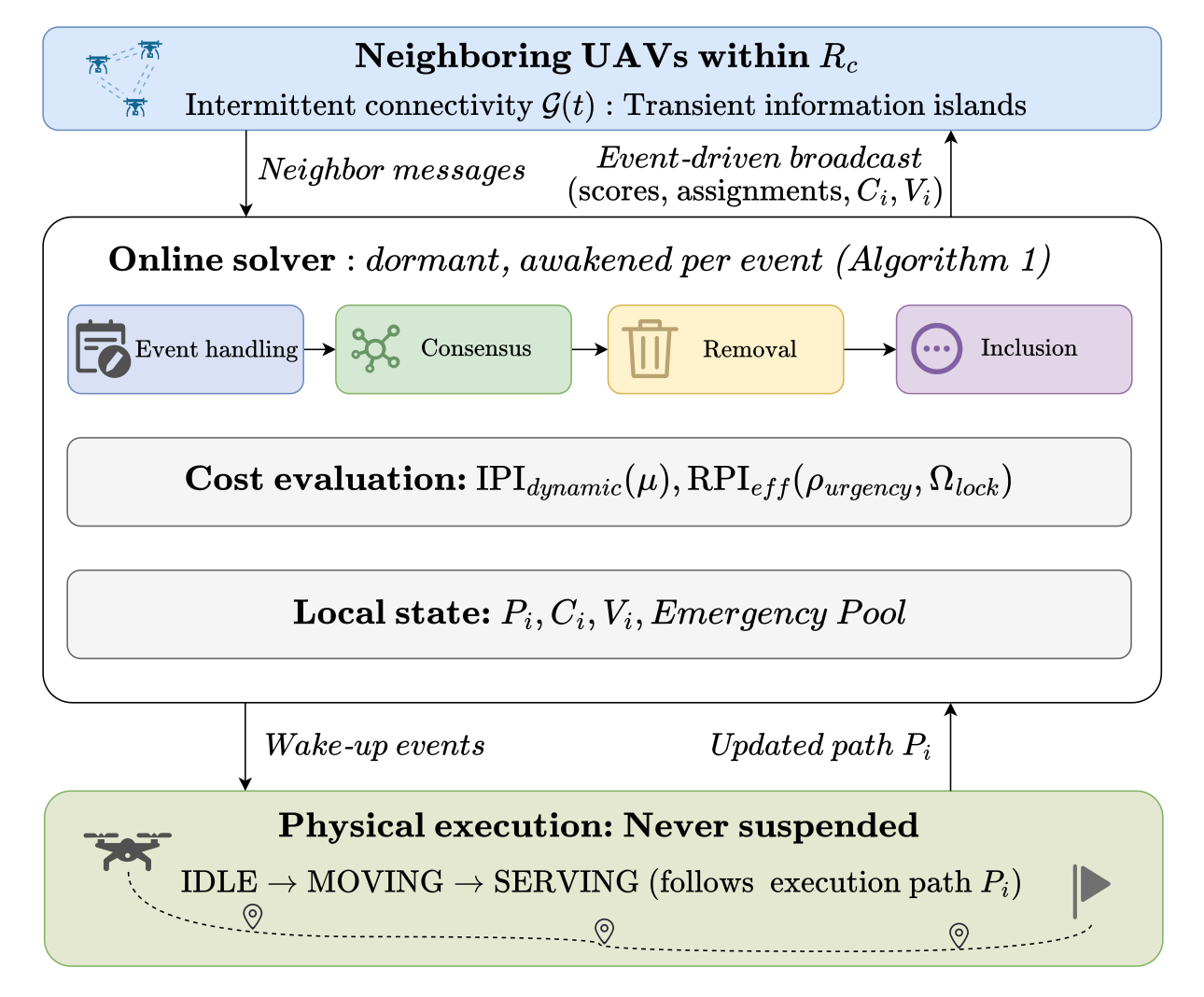}
\caption{The CC-OPI architecture at a single UAV.}
\label{fig:arch}
\end{figure}

To limit computation and communication overhead during prolonged isolation, the core evaluation engine of CC-OPI remains dormant by default and is awakened only by one of three physical or temporal triggers:

\begin{itemize}
    \item \textit{Topology Change Event}: triggered when relative motion among UAVs changes the local neighbor set, i.e., when a neighbor enters or leaves the communication radius $R_c$. It opportunistically activates local state synchronization and conflict resolution between previously disconnected subnets.
    \item \textit{Task Completion Event}: activated when a UAV finishes the service phase of a task or, upon arrival, physically confirms that the task has already been served. It frees the UAV's capacity, moves the task into the locally confirmed completed set, and prompts the UAV to evaluate new candidate tasks from the unallocated pool.
    \item \textit{Heartbeat Event}: triggered periodically at a low, fixed frequency by the passage of physical time. It scans the temporal state (e.g., the remaining time window) of every uncompleted task in the UAV's view, identifying tasks that approach their hard deadlines so that timely emergency intervention is possible.
\end{itemize}

By restricting algorithm activation to these explicit events, CC-OPI consumes computation and communication resources---the latter measured as discrete broadcast messages---only during meaningful state transitions, which suits operation under intermittent connectivity.

Algorithm~\ref{alg:ccopi} summarizes the per-UAV execution flow. Each wake-up runs a four-phase pipeline---event handling, consensus, removal, and inclusion---in which no phase iterates to a fixed point, a property that underpins the bounded per-invocation cost established in Section~\ref{sec:theoretical}. The cost-evaluation metrics and fault-tolerance mechanisms invoked by the pipeline are detailed in the remainder of this section.

\begin{algorithm}[!t]
\caption{CC-OPI Online Execution at UAV $u_i$.}
\label{alg:ccopi}
\small
\begin{algorithmic}[1]
\REQUIRE execution path $P_i \leftarrow [\,]$; score, assignment, and logical-time vectors; completed set $\mathcal{C}_i \leftarrow \emptyset$ with version $V_i \leftarrow 0$; \textit{Emergency Pool} $\leftarrow \emptyset$
\FOR{each physical frame of duration $\Delta t$}
    \STATE advance the physical state machine: fly toward the head of $P_i$, continue service, or idle
    \STATE collect the local events $E$ fired this frame: topology change, task completion, heartbeat, or new-task announcement (when runtime task arrival is enabled)
    \STATE broadcast the local state to neighbors within $R_c$ if $E \neq \emptyset$ or the state changed since the last broadcast
    \IF{$E \neq \emptyset$ \OR a neighbor message arrived}
        \STATE \textit{Phase 1 (event handling):}
        \STATE add each task completed by $u_i$, or observed completed upon arrival, to $\mathcal{C}_i$; advance $V_i$; evict members of $\mathcal{C}_i$ from $P_i$
        \IF{heartbeat $\in E$}
            \STATE place every uncompleted task with $0 < t_{rem}(\tau_k) < T_{emergency}$ in the \textit{Emergency Pool}
        \ENDIF
        \STATE \textit{Phase 2 (consensus):}
        \FOR{each received message}
            \STATE merge the sender's completions missing from $\mathcal{C}_i$, advance $V_i$, and evict the merged tasks from $P_i$
            \STATE resolve task ownership by the standard PI consensus rules \cite{choi2009consensus, zhao2015heuristic}, skipping tasks in $\mathcal{C}_i$
        \ENDFOR
        \STATE \textit{Phase 3 (removal):}
        \FOR{each $\tau_k \in P_i$ whose consensus owner is no longer $u_i$}
            \IF{$\tau_k \in \Omega_{lock}$}
                \STATE reclaim ownership of $\tau_k$
            \ELSIF{$RPI_{raw}(\tau_k)$ exceeds the winning consensus score}
                \STATE remove $\tau_k$ from $P_i$
            \ELSE
                \STATE retain $\tau_k$ and reassert ownership
            \ENDIF
        \ENDFOR
        \STATE \textit{Phase 4 (inclusion):}
        \WHILE{$|P_i| < C_i$}
            \STATE for each type-compatible candidate $\tau_k \notin P_i \cup \mathcal{C}_i$, compute $IPI_{dynamic}(\tau_k)$ over the deadline-feasible insertion positions, never ahead of a locked head; a flagged unattended task is exempt from the spatial penalty ($\mu(\tau_k) = 1$)
            \STATE let the gain of $\tau_k$ be $RPI_{global}(\tau_k) - IPI_{dynamic}(\tau_k)$, adding $\mathcal{A}_{bonus}$ to $RPI_{global}(\tau_k)$ first if $\tau_k$ is flagged and unattended
            \STATE insert the candidate with the largest positive gain at its best position; \textbf{break} if no gain is positive
        \ENDWHILE
        \STATE update $RPI_{eff}(\tau_k)$ for every $\tau_k \in P_i$ and write it into the score vector
    \ENDIF
\ENDFOR
\end{algorithmic}
\end{algorithm}

\subsection{Spatial-Aware Dynamic IPI Evaluation}
\label{subsec:spatial_ipi}
In the standard PI framework, the inclusion of a new task is evaluated by its IPI \cite{zhao2015heuristic}. Let $P_i \oplus_n \tau_k$ denote the insertion of task $\tau_k$ at the $n$-th position of UAV $u_i$'s execution path $P_i$. The raw IPI, $IPI_{raw}(\tau_k)$, is the minimum objective cost increment over all insertion positions that satisfy the time-window constraints:
\begin{equation}
IPI_{raw}(\tau_k) = \min_{n} \left[ \mathcal{F}_2(P_i \oplus_n \tau_k) - \mathcal{F}_2(P_i) \right],
\end{equation}
where $\mathcal{F}_2(\cdot)$ is the execution cost function (the MinSum objective evaluated at the current physical time).

Applying $IPI_{raw}(\tau_k)$ directly in fragmented networks, however, has a clear limitation: without global state information, UAVs in disconnected subnets bid greedily on the same remote, high-value tasks, producing spatial conflicts and wasted travel once the subnets reconnect and consensus is triggered.

To suppress such blind cross-regional movement, CC-OPI introduces a spatial locality penalty factor $\mu(\tau_k)$ that gives UAVs an implicit-coordination capability. The factor scales the insertion cost according to how far a candidate task lies from the UAV's initial spatial anchor $p_{i, init}$:
\begin{equation}
\mu(\tau_k) =
\begin{cases}
1, & \text{unattended emergency task,} \\[2pt]
1 + \lambda \cdot \bar{d}(\tau_k)^\gamma, & \text{otherwise,}
\end{cases}
\end{equation}
where $\bar{d}(\tau_k) = \frac{\|p_{\tau_k} - p_{i, init}\|}{D_{map}}$ is the Euclidean distance normalized by the map diagonal $D_{map}$, and the parameters $\lambda > 0$ and $\gamma \ge 1$ control the strength and growth rate of the penalty, respectively. The anchor is fixed at each UAV's deployment position rather than its current location: a fixed anchor grants every UAV a stable spatial territory, whereas an anchor that drifted with the UAV would continually re-center the low-cost region on its current location and thus fail to curb long-range roaming. The first case lifts the penalty entirely for a task that has been flagged as time-critical and left unattended; this exemption, rather than any change to the task's bid, is the mechanism that permits a UAV to cross regions for a rescue, and its precise condition is given with the emergency-rescue protocol in Section~\ref{subsubsec:emergency}. The normalized anchor distance is speed-agnostic: the same $\bar{d}(\tau_k)$ corresponds to different flight times for the two cruising speeds, a known simplification that we retain for uniformity across the heterogeneous swarm. The fixed anchor likewise presumes a dispersed deployment: if the initial positions are clustered, or unrelated to the task geography, the induced territories overlap and the penalty loses its partitioning effect---a stress case we do not evaluate and leave to future work.

The spatial-aware dynamic IPI used in CC-OPI is then:
\begin{equation}
IPI_{dynamic}(\tau_k) = IPI_{raw}(\tau_k) \times \mu(\tau_k).
\end{equation}

When a UAV is disconnected from the global network, $\mu(\tau_k)$ raises the perceived cost of distant assignments and thus discourages the UAV from claiming remote tasks. This constrains each UAV's effective operating radius and leads the swarm to self-organize into a regionalized operating pattern without explicit communication. Because each UAV concentrates on its local vicinity, overlapping claims become less likely, and fewer assignment collisions remain to be resolved once communication is restored.

\subsection{Urgency-Aware Dynamic RPI}
Complementing task inclusion, the standard PI framework uses the RPI \cite{zhao2015heuristic} to quantify the marginal contribution of an allocated task to a UAV's execution path. For a task $\tau_k$ already in $P_i$, the raw RPI, $RPI_{raw}(\tau_k)$, is the reduction in objective cost when the task is removed:
\begin{equation}
RPI_{raw}(\tau_k) = \mathcal{F}_2(P_i) - \mathcal{F}_2(P_i \setminus \tau_k).
\end{equation}

In traditional consensus protocols, conflicts that arise when subnets reconnect are resolved by comparing the $RPI_{raw}$ of competing UAVs. Once execution and allocation proceed concurrently, this purely cost-based comparison overlooks two factors and can cause assignment oscillation. First, it disregards physical commitment: a UAV in its final approach to a target, or already serving it, may lose ownership to a neighbor with a marginally lower instantaneous bid, interrupting an effectively irreversible action and wasting the committed travel. Second, it disregards temporal urgency: a task near its deadline is treated no differently from a non-urgent one, so it may be preempted by a competitor with a slightly lower bid and subsequently miss its time window.

To address both factors and stabilize assignments during reconnection, CC-OPI defines the effective RPI, $RPI_{eff}(\tau_k)$, as a piecewise function:
\begin{equation}
    RPI_{eff}(\tau_k) =
    \begin{cases}
    -\infty, & \text{if } \tau_k \in \Omega_{lock}, \\
    +\infty, & \text{if } t_{slack} \le 0, \\
    RPI_{raw}(\tau_k) - \rho_{urgency}, & \text{otherwise}.
    \end{cases}
\end{equation}
Here, $\Omega_{lock}$ is the set of tasks granted non-preemptive immunity: if a UAV is actively serving the task at the head of its path, or is approaching that head task with remaining distance below a threshold $d_{lock}$ (e.g., $50$~m), the task is hard-locked to preserve the continuity of the ongoing physical motion. An effective RPI of $-\infty$ is the strongest possible retention bid, so a locked task cannot be preempted in consensus; this protects the irreversible final approach and the in-progress service from disruption. Conversely, if a local evaluation determines that a task will inevitably miss its time window---that is, its remaining time slack $t_{slack} = D_k - A_{i,k}$ satisfies $t_{slack} \le 0$---the UAV broadcasts an RPI of $+\infty$ to signal that it is abandoning ownership.

For any task that is neither locked nor certain to miss its deadline, the effective RPI adjusts the raw value by an urgency penalty $\rho_{urgency}$ that prioritizes tasks close to their deadline. Under a cost-minimization objective such as MinSum, a lower RPI represents a stronger retention bid, so subtracting $\rho_{urgency}$ makes a more urgent task harder for neighbors to preempt. Let $t_{safe}$ be a safety threshold on the remaining slack. The penalty is defined as:
\begin{equation}
    \rho_{urgency} =
    \begin{cases}
    \beta \times \left( \frac{t_{safe} - t_{slack}}{t_{slack}} \right)^{\kappa}, & \text{if } 0 < t_{slack} < t_{safe}, \\
    0, & \text{if } t_{slack} \ge t_{safe},
    \end{cases}
\end{equation}
where $\beta > 0$ scales the magnitude of the penalty and $\kappa \ge 1$ controls how sharply it grows as $t_{slack} \to 0$. When $t_{slack} \ge t_{safe}$, the task is not yet time-critical and incurs no penalty; once $t_{slack}$ drops below $t_{safe}$, $\rho_{urgency}$ rises steeply, progressively shielding the near-deadline task from reassignment.

\subsection{Decentralized State Synchronization and Emergency Rescue}
In fragmented networks, stale information is a primary cause of coordination failure. Disconnected communication subgraphs may persist for long periods, producing large discrepancies in how UAVs perceive global task progress. To maintain eventual consistency without a centralized database, CC-OPI combines a version-based completion-synchronization mechanism with a global-time-driven emergency-rescue protocol. The fault tolerance provided at this layer is state reconciliation---recovery from inconsistent ownership and completion views---rather than tolerance of vehicle hardware failures, which we do not model.

\subsubsection{Version-Based Completion Synchronization}

To prevent a UAV from pursuing a task already completed by peers in other partitions---a situation that can induce coordination deadlock---each UAV $u_i$ maintains a locally confirmed completed-task set $\mathcal{C}_i$ and a monotonically increasing version number $V_i$. A task enters $\mathcal{C}_i$ locally when $u_i$ completes it, or when $u_i$ arrives at the task location and physically observes that it has been served. When a \textit{Topology Change Event} brings UAVs $u_i$ and $u_j$ into communication range, each merges the completions it has not yet recorded, $\mathcal{C}_i \leftarrow \mathcal{C}_i \cup \mathcal{C}_j$, and advances its version counter by the number of newly recorded tasks. The merge is driven by set content rather than by comparing version numbers across UAVs---local counters are not mutually ordered, so a version-gated merge could stall between UAVs holding equal versions but different sets; the version number instead acts as a local change counter that triggers the rebroadcast of the updated state. Each broadcast carries the UAV's full completed set---the scheme is thus a versioned full-state gossip with a monotone merge, not a delta transmission---and completion information propagates through opportunistic encounters in the manner of an epidemic protocol \cite{demers1987epidemic, vahdat2000epidemic}. Any task in $\mathcal{C}_i$ is removed from the UAV's execution path and excluded from the subsequent inclusion phase, which purges obsolete assignments and resolves deadlocks caused by stale information.

\subsubsection{Emergency Pool and Attractiveness Bonus}
\label{subsubsec:emergency}
Beyond consistency, the system must mitigate the risk of task omission caused by spatial regionalization. Driven by the \textit{Heartbeat Event}, each UAV periodically scans the temporal state of every uncompleted task in its local view using the objective remaining time $t_{rem}(\tau_k) = D_k - t$, where $t$ is the current physical time. Unlike the UAV-specific slack $t_{slack} = D_k - A_{i,k}$, this quantity is identical across all UAVs, so they agree on which tasks are at risk. If $0 < t_{rem}(\tau_k) < T_{emergency}$, the task is flagged and placed in a local \textit{Emergency Pool}. This agreement presupposes a globally known task catalog and a shared mission clock---assumptions stated explicitly in Section~\ref{sec:simulation}; since both the heartbeat period and the deadline scale exceed realistic time-synchronization errors by orders of magnitude, a bounded clock skew does not alter which tasks are flagged.

A task in the \textit{Emergency Pool} that is \emph{unattended}---either unassigned or abandoned by its previous owner (signaled by a broadcast $RPI_{eff} = +\infty$)---triggers a two-part response that suspends the spatial regionalization for that task alone.

First, the spatial locality penalty is lifted. As defined in Section~\ref{subsec:spatial_ipi}, the candidate's penalty factor is set to $\mu(\tau_k) = 1$, so its dynamic insertion cost reduces to $IPI_{dynamic}(\tau_k) = IPI_{raw}(\tau_k)$ and is evaluated on its true merit rather than inflated by distance. This exemption is the primary mechanism that allows a UAV with spare capacity to leave its region and perform a cross-regional rescue: the geographic barrier that $\mu(\tau_k)$ erects in the normal case is precisely what would otherwise cause a distant emergency task to be neglected.

Second, the task's priority is raised during inclusion. The algorithm adds a large positive constant $\mathcal{A}_{bonus}$ to its value in the global consensus vector, $RPI_{global}(\tau_k)$. Recall that the inclusion phase adopts candidates that maximize the gain $RPI_{global}(\tau_k) - IPI_{dynamic}(\tau_k)$, so a higher $RPI_{global}$ makes a task more attractive to adopt---the complement of the consensus phase, where a lower RPI denotes a stronger claim to retain a task already held. The bonus is applied as
\begin{equation}
RPI_{global}(\tau_k) \leftarrow RPI_{global}(\tau_k) + \mathcal{A}_{bonus},
\end{equation}
which promotes the flagged task in the inclusion ranking when a UAV's capacity is contested. It acts as a soft priority margin whose exact magnitude is immaterial provided it is large enough to favor the flagged task, a property examined in Section~\ref{subsec:hyperparam}. Together, the exemption and the bonus lead UAVs with spare capacity to rescue tasks under severe time pressure, improving the overall task completion rate.

\section{Theoretical Analysis}
\label{sec:theoretical}

\subsection{Finite-Time Termination}
\label{subsec:convergence}

Classical convergence guarantees for distributed allocation methods such as CBBA and PI \cite{choi2009consensus, zhao2015heuristic} rest on two premises: the task scores are static constants, and the communication network remains sufficiently connected throughout negotiation, mirroring the connectivity conditions under which consensus dynamics over switching topologies are known to converge \cite{olfati2004consensus}. CC-OPI satisfies neither. Its effective scores vary with the UAVs' physical positions and the elapsed time, and its topology is intermittently connected and may fragment into isolated components. We therefore establish \emph{finite-time termination together with deadlock- and livelock-freedom}, rather than convergence to a global optimum, which is unattainable under persistent partitioning. Concretely, three layered facts are proved: each solver invocation halts after a bounded amount of work (Lemma~\ref{lem:bounded}); the mission performs finitely many operations and terminates within a bounded horizon (Theorem~\ref{thm:termination}); and, within that horizon, no stale-completion deadlock and no unbounded reassignment sequence occur (Corollaries~\ref{cor:deadlock} and~\ref{cor:livelock}). We do not claim that the assignment reaches a network-wide fixed point before the task deadlines; under persistent partitioning, no algorithm can guarantee this. The analysis requires no assumption on the marginal-gain structure of the score functions.

\begin{assumption}
\label{asm:standing}
The following standing conditions hold throughout:
\begin{enumerate}
    \item[(A1)] The number of UAVs $N$ and the number of tasks $M$ are finite.
    \item[(A2)] Each cruising speed is bounded below, $v_i \ge v_{min} > 0$; each service time $s_k$ and each capacity $C_i$ is finite.
    \item[(A3)] Each deadline $D_k$ is finite.
    \item[(A4)] The lock distance satisfies $d_{lock} > 0$, and physical time advances in fixed steps $\Delta t > 0$.
    \item[(A5)] Ties in the consensus rules are broken by a fixed total order on the UAV indices.
\end{enumerate}
\end{assumption}

\begin{definition}
\label{def:terms}
A task is \emph{live} at time $t$ if it is neither completed nor expired, where a task is expired once $t > D_k$. The system is in a \emph{terminal state} if every task is completed or expired and every UAV is idle with an empty path. The potential $\Phi(t)$ denotes the number of live tasks, so that $\Phi(t) \in \{0, 1, \dots, M\}$. A \emph{coordination deadlock} is the indefinite holding or pursuit of an already-completed task, sustained by stale completion information; a \emph{reassignment livelock} is an unbounded sequence of ownership transfers of the same task within the mission.
\end{definition}

We recall from Section~\ref{sec:cc-opi} that a task belongs to the lock set $\Omega_{lock}$ when it is the head of its owner's path and the owner is serving it or moving toward it with remaining distance below $d_{lock}$, in which case its effective removal score is $RPI_{eff} = -\infty$.

\begin{lemma}[Bounded computation per invocation]
\label{lem:bounded}
Each wake-up of the CC-OPI online solver processes the pending events and then executes the consensus, removal, and inclusion phases of Section~\ref{sec:cc-opi} exactly once, halting after a number of operations that is finite and polynomial in $N$, $M$, and $C_i$. The explicit bound is given in Section~\ref{subsec:complexity}.
\end{lemma}

\begin{IEEEproof}
No phase is iterated to a fixed point. Event handling performs a single pass over the task set; the consensus phase performs a single pass over the received neighbor messages, scanning the task set once per message; the removal phase performs a single pass over the current path, whose length is at most $C_i$; and the inclusion phase admits at most $C_i$ tasks, each admission requiring a scan of at most $M$ candidates over at most $C_i$ insertion positions. Each phase thus processes a finite set a bounded number of times, so the invocation performs finitely many operations and returns.
\end{IEEEproof}

\begin{lemma}[Monotone confirmed completion]
\label{lem:monotone}
For each UAV $u_i$, the locally confirmed completed set $\mathcal{C}_i(t)$ is non-decreasing in $t$ and satisfies $|\mathcal{C}_i(t)| \le M$.
\end{lemma}

\begin{IEEEproof}
The set $\mathcal{C}_i$ is modified only by union: with a task completed by $u_i$ or observed completed upon arrival, or with the previously unrecorded elements of a neighbor's completed set during a topology-change synchronization. No element is ever removed, so $|\mathcal{C}_i|$ is non-decreasing; the bound $|\mathcal{C}_i| \le M$ follows from (A1).
\end{IEEEproof}

\begin{corollary}[Deadlock-freedom within a component]
\label{cor:deadlock}
Within any connected component whose members remain connected until the completed-set synchronization finishes, no completed task is re-included by its members; hence no coordination deadlock arises from stale completion information inside such a component.
\end{corollary}

\begin{IEEEproof}[Proof sketch]
The completion fact propagates monotonically (Lemma~\ref{lem:monotone}) to every member that remains in the component within finitely many synchronization rounds---and, after a split, along any later temporal contact path (the gossip semantics of Section~\ref{sec:cc-opi})---whereupon each holder evicts the task and permanently excludes it from the inclusion phase. The complete proof is given in the supplementary material.
\end{IEEEproof}

\begin{lemma}[Eventual commitment]
\label{lem:commit}
Let $u_i$ move toward its primary target $\tau$, the head of its path. Once the remaining distance falls below $d_{lock}$, $\tau \in \Omega_{lock}$ and $RPI_{eff}(\tau) = -\infty$; thereafter $\tau$ is neither preempted through consensus nor displaced from the head of the path. The task $\tau$ then leaves the path of $u_i$ within $d_{lock}/v_{min} + s_\tau$, either served by $u_i$ or, if another UAV is already serving it, dropped by $u_i$ upon arrival. Moreover, if $\tau$ is deadline-feasible at the instant it is locked, this resolution is an on-time completion by exactly one UAV.
\end{lemma}

\begin{IEEEproof}[Proof sketch]
A locked task carries the strongest possible retention bid ($RPI_{eff} = -\infty$), is reclaimed rather than released by the removal phase, and cannot be displaced from the head of the path by insertion, so it remains the head until physically resolved; resolution within $d_{lock}/v_{min} + s_\tau$ follows from the bounded remaining distance and the non-interruptible service. The complete proof, including the treatment of concurrent locking and of heads that are already deadline-infeasible when locked, is given in the supplementary material.
\end{IEEEproof}

\begin{theorem}[Finite-time termination]
\label{thm:termination}
Under Assumption~\ref{asm:standing}, CC-OPI reaches a terminal state within finite physical time and performs finitely many operations. Recalling the map diagonal $D_{map}$ from Section~\ref{sec:cc-opi}, the termination time satisfies
\begin{equation}
\label{eq:tstop}
    T_{stop} \le T_{max} := \max_k D_k + \Bigl(\max_i C_i\Bigr)\!\left(\frac{D_{map}}{v_{min}} + \max_k s_k\right).
\end{equation}
\end{theorem}

\begin{IEEEproof}[Proof sketch]
\emph{(i)} Once $t > \max_k D_k$ no task is live, so the inclusion phase admits nothing and each UAV merely drains its path of length at most $\max_i C_i$, each head resolving within $D_{map}/v_{min} + \max_k s_k$ (Lemma~\ref{lem:commit}), which yields~\eqref{eq:tstop}. \emph{(ii)} By (A4) the number of frames is at most $T_{max}/\Delta t$, and each wake-up is bounded by Lemma~\ref{lem:bounded}. \emph{(iii)} Resolution is permanent, so $\Phi(t)$ is non-increasing; Corollary~\ref{cor:deadlock} precludes stale-completion deadlock, and Lemma~\ref{lem:commit} ensures that committed progress is resolved rather than reversed. The complete proof is given in the supplementary material.
\end{IEEEproof}

The horizon bound~\eqref{eq:tstop} follows from the finite deadlines together with the bounded time to drain each path; the substance of Theorem~\ref{thm:termination} lies in parts (ii) and (iii), which guarantee that within this horizon the algorithm performs only finitely many operations and makes monotone progress toward termination, rather than expending the horizon on unbounded negotiation or reassignment oscillation.

\begin{corollary}[Deadlock- and livelock-freedom]
\label{cor:livelock}
Within the execution horizon of Theorem~\ref{thm:termination}, CC-OPI admits neither coordination deadlock arising from stale completion information nor an unbounded reassignment sequence.
\end{corollary}

\begin{IEEEproof}[Proof sketch]
Deadlock-freedom follows from Corollary~\ref{cor:deadlock} together with Remark~\ref{rem:scope}; livelock-freedom follows because the wake-up budget over $[0, T_{max}]$ is finite (Theorem~\ref{thm:termination}), each wake-up modifies a path by a bounded amount (Lemma~\ref{lem:bounded}), and a locked commitment is irreversible (Lemma~\ref{lem:commit}). The complete proof is given in the supplementary material.
\end{IEEEproof}

\begin{remark}[Scope of the guarantee]
\label{rem:scope}
The guarantee established here is finite-time termination together with deadlock- and livelock-freedom; it is not convergence to a global optimum, which is unattainable under persistent partitioning. In particular, two UAVs in disconnected components may transiently hold the same task; such a conflict is resolved by the consensus rules once the components merge. Components that never regain contact cannot resolve such a conflict---an impossibility of coordination without any communication opportunity, not a deficiency of the algorithm.
\end{remark}

\subsection{Complexity Analysis}
\label{subsec:complexity}

We now quantify the per-UAV computational, communication, and storage costs of CC-OPI. Let $N$ be the number of UAVs, $M$ the number of tasks, and $C := \max_i C_i$ the maximum path length, which is bounded by the payload capacity. A central property of the event-driven design is that these costs are incurred only at discrete wake-ups, rather than continuously.

\subsubsection{Computational Complexity}
Consider a single wake-up of UAV $u_i$ (Lemma~\ref{lem:bounded}). Scoring a candidate path of length at most $C$ requires $O(C)$ time. The costs of the three phases are as follows.
\begin{itemize}
    \item \emph{Consensus} compares the local state with each of at most $N-1$ neighbor messages over the $M$ tasks, giving $O(NM)$.
    \item \emph{Removal} performs a single pass over the at most $C$ held tasks, each evaluated by an $O(C)$ path-cost recomputation, giving $O(C^2)$.
    \item \emph{Inclusion} admits at most $C$ tasks; each admission scans at most $M$ candidates, and evaluating one candidate searches at most $C$ insertion positions with an $O(C)$ feasibility-and-cost check per position. This dominates, at $O(M C^3)$.
\end{itemize}
The cost of one wake-up is therefore
\begin{equation}
\label{eq:comp_wakeup}
    O\!\left(M C^3 + N M\right).
\end{equation}
Because the capacity $C$ is a small constant in practice, \eqref{eq:comp_wakeup} reduces to $O(NM)$ per wake-up, matching the per-iteration cost of the underlying PI consensus. Over the bounded horizon of Theorem~\ref{thm:termination}, the number of wake-ups is at most $N \lceil T_{max}/\Delta t \rceil$, so the mission-level computation is finite and proportional to the number of simulated frames. We state the character of this bound precisely: the per-wake-up cost is polynomial in the discrete dimensions $N$, $M$, and $C$, whereas the frame count $T_{max}/\Delta t$ scales with the numerical magnitudes of the deadlines, the map diameter, the inverse speed, and the inverse time step---a pseudo-polynomial dependence unless these physical quantities are treated as fixed constants.

\subsubsection{Communication and Storage Complexity}
At each broadcast, a UAV transmits its score vector ($O(M)$), its assignment vector ($O(M)$), its logical-time vector ($O(N)$), and its versioned completed set ($O(M)$). The message size is thus
\begin{equation}
\label{eq:msg_size}
    O(M + N).
\end{equation}
A UAV broadcasts only when an event fires or its local state changes substantively, rather than once per frame. Let $E$ denote the total number of such broadcasts over a run. The total communication volume is then $O\!\left(E (M + N)\right)$, with the worst-case bound $E \le N \lceil T_{max}/\Delta t \rceil$ stated in Proposition~\ref{prop:comm}.

\begin{proposition}[Communication-overhead bound]
\label{prop:comm}
Over the horizon $[0, T_{max}]$ of Theorem~\ref{thm:termination}, the number of broadcasts satisfies $E \le N \lceil T_{max}/\Delta t \rceil$, and the total communication volume is $O\!\left(N \lceil T_{max}/\Delta t \rceil (M+N)\right)$.
\end{proposition}

\begin{IEEEproof}
By (A4), each UAV broadcasts at most once per frame and there are at most $\lceil T_{max}/\Delta t \rceil$ frames over $[0, T_{max}]$; summing over the $N$ UAVs bounds $E$, and multiplying by \eqref{eq:msg_size} bounds the volume.
\end{IEEEproof}

The bound in Proposition~\ref{prop:comm} coincides with that of a synchronous round-based scheme only in the degenerate case where an event fires at every UAV in every frame. Under intermittent connectivity, events are sparse: broadcasts are triggered by topology changes, completions, and periodic heartbeats, so $E$ scales with the number of such events rather than with the number of frames. This separation between worst-case and typical communication cost is the principal advantage of the event-driven paradigm over frame-synchronous polling; the measured communication cost is reported in Section~\ref{sec:simulation}. Finally, the per-UAV storage---the score, assignment, and logical-time vectors, the completed set, the emergency pool, and the execution path---is likewise $O(M + N)$, and no UAV stores any global structure indexed by both $N$ and $M$ simultaneously.

\section{Simulation and Experimental Results}
\label{sec:simulation}

This section evaluates CC-OPI through extensive simulations. We first describe the simulation setup, the evaluation protocol, and the performance metrics (Section~\ref{subsec:setup}). We then report a four-regime performance comparison against the classical and matched online baselines (Section~\ref{subsec:main_comparison}), an ablation study that isolates the contribution of each mechanism (Section~\ref{subsec:ablation}), sensitivity analyses with respect to the algorithm's hyperparameters (Section~\ref{subsec:hyperparam}) and to the communication radius $R_c$ (Section~\ref{subsec:radius}), and a robustness study under link-level packet loss, terrain occlusion, and dynamic task arrival (Section~\ref{subsec:robustness}). All reported results are averaged over $30$ independent random seeds. Tables report the mean and standard deviation; error bars or shaded bands in the figures denote $t$-based $95\%$ confidence intervals of the mean.

\subsection{Simulation Setup and Metrics}
\label{subsec:setup}

\subsubsection{Simulation Setup}
We simulate a post-disaster search and rescue scenario on a $10{,}000\times10{,}000$~m field. The UAV swarm is heterogeneous, comprising medical-supply UAVs (cruising speed $30$~m/s) and food-supply UAVs (cruising speed $50$~m/s) in a $1{:}1$ ratio; each UAV has a fixed payload capacity $C_i$ and a communication radius $R_c=250$~m. This radius is representative of the effective air-to-air range measured for UAV ad hoc networks with commodity radios \cite{hayat2016survey, asadpour2013characterizing}; its value is swept as an explicit variable in Section~\ref{subsec:radius}. In graph terms it is a severe connectivity regime: at $R_c=250$~m and the representative scale, the initial communication graph is typically empty, the time-averaged mean degree is $0.16$ with the largest connected component covering $15\%$ of the swarm, and a mission produces about $15$ pairwise contacts of median duration $12$~s; the corresponding statistics for every swept radius are tabulated in the supplementary material, so that the results can be mapped to other environments. Tasks are likewise split evenly between the two supply types, with service times of $300$~s (medical) and $350$~s (food), and each task is assigned a hard deadline drawn uniformly from a fixed range. The physical engine advances in steps of $\Delta t=1$~s, and the heartbeat that drives the global-time scan fires every $20$~s. In-range messages are delivered reliably within the same simulation step; the loss assumption is relaxed in the robustness study of Section~\ref{subsec:robustness}. UAVs, tasks, and deadlines are randomly generated per seed within these bounds. The principal parameters are summarized in Table~\ref{tab:setup}, and Table~\ref{tab:assumptions} consolidates the information and physical assumptions of the model; each assumption holds identically for every method and regime evaluated below.

To assess scalability, we vary the problem size across six levels, with the UAV count $N\in\{6,8,10,12,14,16\}$ and the task count $M$ fixed by the task-to-UAV ratio. Unless otherwise noted, two settings are used. The \emph{default setting} fixes the ratio at $M=2N$, the capacity at $C_i=3$, and the deadline range at $[0,2000]$~s. Under this setting, the task completion rate of the stronger methods approaches saturation, which compresses the differences between them. We therefore adopt a \emph{tight-constraint setting} as the primary scenario for the performance comparison: the load is raised to $M=3N$, the deadlines are tightened to $[0,1200]$~s, and the capacity is increased to $C_i=4$. We deliberately keep the capacity generous so that $C_i\cdot N>M$; the difficulty thus arises from the time windows rather than from a shortage of task slots, which avoids structurally favoring the online paradigm. Under the tight-constraint setting, the full-connectivity results drop to roughly $0.80$--$0.93$ depending on the method, so the metrics are no longer saturated and the differences between methods become measurable. A dedicated scaling study additionally runs CC-OPI at up to $128$ UAVs and $384$ tasks, reporting runtime, per-wake-up latency, memory, and message sizes (supplementary material).

\begin{table}[tbp]
\centering
\caption{Principal Simulation Parameters}
\label{tab:setup}
\renewcommand{\arraystretch}{1.2}
\begin{tabular}{lll}
\toprule
\textbf{Parameter} & \textbf{Default} & \textbf{Tight-constraint} \\
\midrule
Field size & \multicolumn{2}{c}{$10{,}000\times10{,}000$~m} \\
Communication radius $R_c$ & \multicolumn{2}{c}{$250$~m} \\
UAV speed (medical / food) & \multicolumn{2}{c}{$30$ / $50$~m/s} \\
Task service time (med. / food) & \multicolumn{2}{c}{$300$ / $350$~s} \\
Time step $\Delta t$ / heartbeat & \multicolumn{2}{c}{$1$~s / $20$~s} \\
\midrule
Task-to-UAV ratio $M{:}N$ & $2{:}1$ & $3{:}1$ \\
Payload capacity $C_i$ & $3$ & $4$ \\
Deadline range & $[0,2000]$~s & $[0,1200]$~s \\
\midrule
UAV count $N$ & \multicolumn{2}{c}{$\{6,8,10,12,14,16\}$} \\
Random seeds & \multicolumn{2}{c}{$30$ (mean $\pm$ std)} \\
\bottomrule
\end{tabular}
\end{table}

\begin{table}[tbp]
\centering
\caption{Information and Physical Modeling Assumptions}
\label{tab:assumptions}
\renewcommand{\arraystretch}{1.2}
\footnotesize
\begin{tabular}{@{}>{\raggedright\arraybackslash}p{0.24\columnwidth}>{\raggedright\arraybackslash}p{0.66\columnwidth}@{}}
\toprule
\textbf{Assumption} & \textbf{Scope and implication} \\
\midrule
Static tasks globally known & Every UAV knows all initial tasks (position, type, deadline) from deployment; only the assignment and completion state is subject to the communication constraints. \\
\addlinespace[2pt]
Injected tasks globally discovered & The existence of a runtime task is announced to all UAVs upon arrival; its allocation remains communication-constrained (Section~\ref{subsec:robustness}). \\
\addlinespace[2pt]
Synchronized clocks & Deadlines and the emergency criterion $t_{rem}(\tau_k)=D_k-t$ presume a shared physical time; clock skew is not modeled. \\
\addlinespace[2pt]
Ideal in-range links & Connectivity is a binary graph determined by the radius $R_c$ and line of sight; in-range messages arrive within one simulation step, without loss, delay, reordering, or bandwidth contention. The loss, delay, and reordering assumptions are relaxed in Section~\ref{subsec:robustness}. \\
\addlinespace[2pt]
Straight-line motion & UAVs fly at constant speed along straight segments; obstacles block communication only, never the flight path (Section~\ref{subsec:robustness}). \\
\addlinespace[2pt]
Unlimited endurance, no failures & No energy budget and no mid-mission UAV failure are modeled; the mission horizon is bounded by the deadlines alone (Section~\ref{sec:theoretical}). Both are future work (Section~\ref{sec:conclusion}). \\
\bottomrule
\end{tabular}
\end{table}

\subsubsection{Evaluation Protocol}
A central methodological point is that communication constraints are an intrinsic property of the target scenario rather than a handicap imposed on the baselines. Accordingly, we evaluate four regimes:
\begin{itemize}
    \item \textbf{(i) Full-connectivity reference:} the classical methods (PI, CBBA, and PI-MaxAss) running on a fully connected network. This regime reflects the idealized assumption under which these methods were designed, and serves as each method's ideal-connectivity reference; it is not an upper bound on the achievable performance (see Section~\ref{subsec:radius}).
    \item \textbf{(ii) Naive-transfer baseline:} the same classical methods running on the radius-limited network, denoted PI-C, CBBA-C, and PI-MaxAss-C. These represent the outcome of directly deploying a static method in the communication-constrained environment.
    \item \textbf{(iii) Matched online baselines:} the unmodified scoring functions and consensus rules of PI and CBBA executed in the same event-driven online regime as CC-OPI, denoted PI-OM and CBBA-OM, with negotiation re-run at every wake-up over the radius-limited network. These baselines separate the benefit of the online paradigm from the benefit of CC-OPI's specific mechanisms.
    \item \textbf{(iv) CC-OPI:} the proposed method on the radius-limited network.
\end{itemize}
For regime (ii), the static ``allocate-then-execute'' paradigm reaches a single consensus over the connectivity available at $t=0$ and then freezes the assignment for execution, without replanning while the topology evolves. This is a faithful reflection of the inherent limitation of the static paradigm---execution-time replanning is precisely the online paradigm that constitutes our contribution---rather than an artificially weakened baseline.

Beyond the four regimes, the main comparison includes a static geographic-partition baseline, denoted GEO, representing the territory-based alternative of Section~\ref{sec:related}: every task is statically owned by the type-compatible UAV whose deployment position is nearest (a Voronoi partition over the pre-mission deployment, shared before launch), each UAV plans its own region by cheapest feasible insertion without any bidding or consensus, completed-task knowledge is identical to CC-OPI's, and a UAV whose own insertable tasks are exhausted may adopt a flagged near-deadline task from any region, duplicate pursuits being resolved by physical confirmation on arrival. GEO runs in the same online paradigm on the radius-limited network, on the same seeds.

For regime (iii), the scoring and consensus logic of the wrapped methods is left untouched; a thin wrapper adds only the adaptations without which any allocation method fails trivially online. Completion knowledge is identical to CC-OPI's: tasks confirmed completed---by the UAV itself, observed on arrival, or learned from neighbors' broadcasts---are evicted from the plan and withdrawn from bidding; a planned task whose estimated arrival time exceeds its deadline is released together with its consensus entry. For CBBA only, the time-discounted bids of the held tasks are re-scored from the current state at each wake-up, since a bid placed early in absolute time could otherwise never be outbid later; PI recomputes its removal scores on every invocation and needs no refresh. All four regimes share the same seeds and problem instances, per-frame topology and communication opportunities, task visibility and deadline semantics, and broadcast-count accounting; in regimes (iii) and (iv), motion is never suspended.

\subsubsection{Evaluation Metrics}
We report the following metrics, computed from the final simulation state:
\begin{itemize}
    \item \emph{Task completion rate} (TCR, $\uparrow$): the fraction of \emph{reachable} tasks that a type-compatible UAV reaches, and starts serving, no later than their deadline and subsequently completes (arrival exactly at the deadline counts as on time; see Section~\ref{problem}). The criterion is enforced directly: the simulator records the first service-start timestamp of every task, and a physically finished service that started late is excluded from all completion-based metrics. Such exclusions are rare---a dedicated audit over the central configurations ($720$ runs) finds $18$ late starts among $15{,}901$ finished services ($0.1\%$), most of them in ablated variants with protection mechanisms disabled. The denominator excludes tasks that are physically impossible to reach in time even under direct flight, so that the metric is not distorted by infeasible instances. TCR is the primary metric for cross-algorithm comparison.
    \item \emph{Makespan} ($\downarrow$): the completion time of the last completed task.
    \item \emph{Total distance} ($\downarrow$): the aggregate flight distance of all UAVs.
    \item \emph{Wasted distance} ($\downarrow$): the flight distance rendered ineffective by reassignment. For the static baselines, allocation precedes any movement, so this quantity is zero by construction; it is incurred only by online replanning.
    \item \emph{Reassignment count} ($\downarrow$): the number of ownership transfers, reflecting coordination overhead and assignment oscillation.
    \item \emph{Messages per completed task} ($\downarrow$): the number of messages normalized by the number of completed tasks, where one message is one event-driven broadcast for the online methods and one per-round state exchange for the static methods; we normalize because the two semantics make absolute counts incomparable. Each broadcast carries a serialized payload of $O(M+N)$ (Section~\ref{sec:theoretical})---at most $16M+8N+8$ bytes under the canonical encoding, i.e., $568$~B at the representative scale, the maximum being reached only when the monotone completed set is full. We additionally record the receiver-side deliveries---one delivery being the receipt of a broadcast by one in-range receiver---and the serialized payload bytes of every message under a canonical binary encoding ($8$-byte floats, $4$-byte integers), so that the communication cost can be compared in sender-side and receiver-side units alike. All communication metrics are application-layer quantities---broadcast attempts, successful deliveries, and canonical payload bytes---and do not measure airtime: over-the-air contention, retransmission, and channel occupation are not modeled.
    \item \emph{Emergency rescue rate} ($\uparrow$): among all tasks that ever entered the emergency pool, the fraction ultimately completed. This metric is meaningful only for CC-OPI, since the baselines maintain no emergency pool; it is therefore used solely as internal evidence in the ablation study and never for cross-algorithm comparison.
\end{itemize}

\subsubsection{Statistical Protocol}
All algorithms and variants are evaluated on the same $30$ problem instances (seeds), so every cross-method comparison is paired by instance. For the main comparison of Section~\ref{subsec:main_comparison}, which we treat as the planned primary analysis, we report per-instance paired differences with $95\%$ confidence intervals obtained from a percentile bootstrap ($10^4$ resamples). The comparison set is small and directionally specified, though not formally preregistered, and the same seeds were also used during development and sensitivity exploration; we therefore rely on effect sizes with confidence intervals rather than dichotomous significance claims and apply no multiplicity correction. The mechanism, sensitivity, and robustness studies (Sections~\ref{subsec:ablation}--\ref{subsec:robustness}) are exploratory. Because the six problem scales sharing a seed are generated from nested random streams, comparisons pooled across scales use a seed-cluster bootstrap that resamples seeds and retains all scales within each draw; the per-scale intervals remain the primary evidence. With $30$ seeds, the standard error of the mean (SEM) of the TCR is about one percentage point, and the paired confidence intervals have half-widths below two percentage points. Finally, the headline comparisons are re-validated on $30$ fresh hold-out seeds that played no role in development or tuning; the effect sizes replicate (supplementary material).

\subsection{Main Performance Comparison}
\label{subsec:main_comparison}

We now compare CC-OPI against the classical and matched online baselines under the four regimes defined in Section~\ref{subsec:setup}, with the tight-constraint setting as the primary scenario. Fig.~\ref{fig:hard_tcr} reports the TCR against the problem scale---with CC-OPI on the radius-limited network drawn as a shared reference line in both panels---and Fig.~\ref{fig:hard_eff} reports the corresponding wasted distance and communication overhead. Table~\ref{tab:hard_summary} summarizes all metrics, including the makespan, at a representative scale of $10$ UAVs and $30$ tasks.

\begin{figure*}[!t]
\centering
\subfloat[Full-connectivity baselines (reference)]{\includegraphics[width=0.44\textwidth]{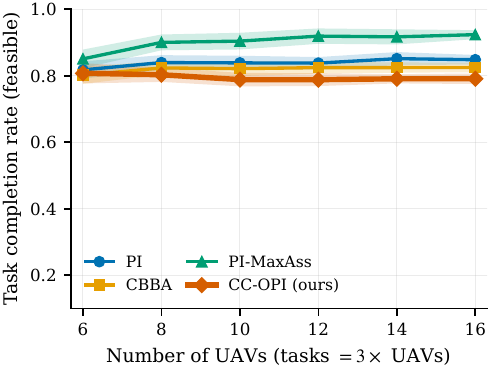}\label{fig:hard_tcr_full}}
\hfil
\subfloat[Constrained baselines (naive transfer)]{\includegraphics[width=0.44\textwidth]{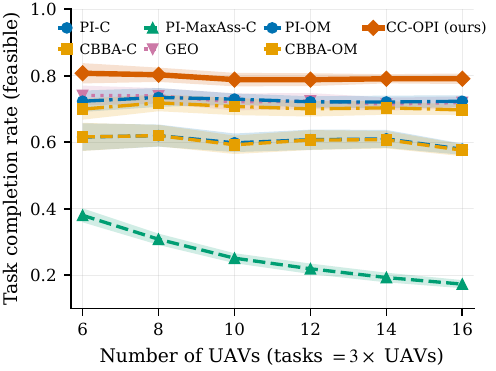}\label{fig:hard_tcr_con}}
\caption{Task completion rate versus problem scale under the tight-constraint setting.}
\label{fig:hard_tcr}
\end{figure*}

\subsubsection{Task Completion Rate}
The TCR results reveal five findings.

\emph{First, the naive transfer of static methods collapses under communication constraints, whereas CC-OPI remains robust.} As shown in Fig.~\ref{fig:hard_tcr_con}, the constrained baselines PI-C and CBBA-C achieve a TCR of only about $0.58$--$0.62$, and their performance further deteriorates as the scale grows. CC-OPI, operating under the same $250$~m communication limit, sustains a TCR of roughly $0.80$ that is essentially flat across scales. The margin over the naive baselines is between $18$ and $22$ percentage points and is widest at the largest scale (paired differences pooled across scales, seed-cluster bootstrap: $+19.0$ percentage points against PI-C, $95\%$ CI $[+17.5, +20.5]$, and $+19.2$ against CBBA-C, $[+17.8, +20.6]$), indicating that CC-OPI's online replanning compensates for the connectivity that the static paradigm assumes but cannot obtain.

\emph{Second, the matched online comparison attributes this margin to two sources: the online paradigm accounts for the larger share, and the CC-OPI mechanisms add a consistent gain on top.} The naive baselines conflate two differences---the execution regime and the allocation logic---so their collapse alone cannot isolate the value of CC-OPI's mechanisms. The matched baselines PI-OM and CBBA-OM remove the first difference. Executing the unmodified PI rules online recovers $+12.1$ percentage points over the frozen PI-C (paired difference pooled across scales, $95\%$ CI $[+10.9, +13.2]$; CBBA-OM over CBBA-C: $+10.1$ $[+9.0, +11.3]$). On top of this paradigm gain, CC-OPI leads PI-OM by a further $+6.9$ points ($[+6.1, +7.8]$ pooled) and CBBA-OM by $+9.0$ ($[+8.1, +10.0]$), and the per-scale differences neither vanish nor shrink as the swarm grows (Fig.~\ref{fig:hard_tcr_con}); the mechanism-level attribution of this gain is given in Section~\ref{subsubsec:backbone}. CBBA-OM also exposes why commitment protection matters: a consensus auction run online without it oscillates severely, with a median of $660$ ownership transfers per run at the representative scale (interquartile range $216$--$1295$, maximum $3186$) against $19$ ($16$--$23$, maximum $29$) for CC-OPI, and roughly twice CC-OPI's messages per completed task (Table~\ref{tab:hard_summary}). CBBA releases its whole bundle suffix on every lost conflict and each released task counts as one transfer, so its absolute count is not directly comparable to the PI-based methods; the order-of-magnitude gap and the doubled message rate nonetheless indicate genuine instability.

\emph{Third, the geographic-partition baseline confirms that the mechanisms add value beyond static territories.} GEO already avoids the frozen-consensus collapse: pooled across scales it leads PI-C by $+12.2$ percentage points ($95\%$ CI $[+11.0, +13.3]$, seed-cluster bootstrap) and is statistically indistinguishable from PI-OM ($+0.1$ $[-0.9, +1.1]$)---spatial decomposition alone recovers roughly the paradigm gain. CC-OPI nevertheless leads GEO by $+6.8$ points pooled ($[+5.9, +7.8]$), with per-scale differences of $+6.4$ to $+7.6$ whose intervals all exclude zero, and incurs $32.6$~km less redundant travel at the representative scale ($[-40.1, -25.2]$; Table~\ref{tab:hard_summary}): static territories cannot rebalance load across regions, and their emergency adoptions duplicate pursuits that CC-OPI resolves through consensus. The value of the PI machinery is thus established against the territory-based alternative, not only against frozen consensus.

\emph{Fourth, CC-OPI approaches the full-connectivity performance of PI and CBBA using only local communication.} Despite having no access to global information, CC-OPI attains a TCR comparable to that of CBBA (paired difference pooled across scales: $-2.5$ percentage points, $95\%$ CI $[-3.3, -1.8]$) and within roughly six percentage points of PI ($-5.0$ $[-6.6, -3.4]$ at the representative scale), both evaluated on a fully connected network (Fig.~\ref{fig:hard_tcr_full}). The remaining gap to PI-MaxAss, the strongest full-connectivity reference (about $0.85$--$0.92$), reflects the price of operating in a decentralized manner without the global picture; we do not claim to match a method that re-optimizes the complete assignment under full connectivity.

\emph{Fifth, the methods that optimize most aggressively for the global picture are the most brittle when that picture is unavailable.} PI-MaxAss is the strongest method under full connectivity (up to about $0.92$) but the weakest under the communication constraint: PI-MaxAss-C degrades from $0.380$ at the smallest scale to $0.173$ at the largest. Its two-stage rescheduling, which presumes a globally consistent view to maximize the number of assigned tasks, becomes counterproductive once that view is fragmented across information islands.

\subsubsection{Efficiency and Overhead}
Fig.~\ref{fig:hard_eff} and Table~\ref{tab:hard_summary} characterize the cost of this robustness. CC-OPI's makespan is comparable to CBBA's and lower than PI-MaxAss's, so the TCR gain does not slow the mission. The cost appears in two coupled quantities: the wasted distance, zero by construction for the static baselines, grows with scale under online reassignment, and the communication overhead is about an order of magnitude above PI or CBBA, because CC-OPI coordinates continuously throughout execution rather than only once beforehand. Against the matched online baselines, however, the comparison reverses: at every scale CC-OPI incurs the least wasted distance and the fewest messages per completed task of the three online consensus methods ($79.6$~km and about $55$ messages at $16$ UAVs, versus $95.5$/$114.8$~km and $73$/$130$ messages for PI-OM/CBBA-OM), so these costs are a property of online execution itself---one that CC-OPI's mechanisms mitigate rather than aggravate. The geographic baseline broadcasts even less, since it exchanges only completion facts, but pays for the missing negotiation with the highest wasted distance of every constrained method at every scale (Fig.~\ref{fig:hard_eff}). CC-OPI's own makespan exceeds the matched baselines' (about $1294$~s versus $1170$--$1180$~s at the representative scale), the counterpart of the emergency-rescue behavior: the additional completions it secures are precisely near-deadline tasks that start service on time but finish service after the deadline. These results should be read as a trade-off rather than an advantage: CC-OPI spends additional communication and some redundant travel to secure its completion gains.

Table~\ref{tab:comm} completes the communication accounting in receiver-side and byte-level units, and the choice of unit matters. In sender-side units---application-layer broadcast attempts and transmitted payload bytes---CC-OPI indeed pays about an order of magnitude more than the full-connectivity methods. The receiver side inverts the picture: on the sparse $250$~m network a CC-OPI broadcast reaches only $0.25$ receivers on average, whereas every full-connectivity exchange is delivered to all $N-1$ peers, so per completed task CC-OPI induces $14.5$ deliveries ($6.5$~kB) against $45.7$ ($19.8$~kB) for PI---about a third of the received traffic of the idealized consensus. The naive constrained baselines mark the opposite degenerate corner: their broadcasts on the typically empty initial graph reach no receiver at all, which quantifies why the frozen transfer fails. Among the three online methods, CC-OPI is the cheapest in every unit; the order-of-magnitude penalty is thus a statement about application-layer broadcast attempts on a mostly empty neighborhood---since channel occupation and contention are not modeled, we draw no conclusion about airtime or congestion.

Placing every method in a single completion-versus-cost plane at the representative scale (Fig.~S1 in the supplementary material) summarizes this trade-off: the full-connectivity methods reach a comparable TCR at roughly an order of magnitude less communication (an advantage derived from the idealized connectivity they assume); the naively transferred baselines fall to the low-completion region; the matched online baselines sit below and to the right of CC-OPI, which dominates them on both axes; and the geographic baseline sits below CC-OPI at a comparable message cost. Among the evaluated methods, none dominates CC-OPI within the constrained regime: reaching its TCR otherwise requires the full-connectivity assumption that the target scenario denies.

\begin{figure}[tbp]
\centering
\includegraphics[width=\columnwidth]{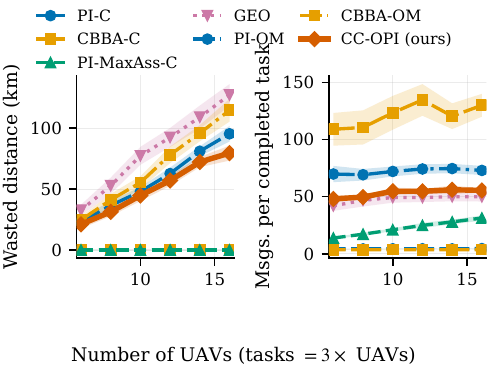}
\caption{Coordination overhead versus problem scale under the tight-constraint setting.}
\label{fig:hard_eff}
\end{figure}

\begin{table*}[!t]
\centering
\caption{Performance Summary at $10$ UAVs / $30$ Tasks (Tight-Constraint Setting)}
\label{tab:hard_summary}
\renewcommand{\arraystretch}{1.25}
\setlength{\tabcolsep}{4.5pt}
\begin{tabular}{llcccccc}
\toprule
\textbf{Regime} & \textbf{Method} & \textbf{TCR}\,$\uparrow$ & \textbf{Makespan (s)}\,$\downarrow$ & \textbf{Total Dist. (km)}\,$\downarrow$ & \textbf{Wasted Dist. (km)}\,$\downarrow$ & \textbf{Reassign.}\,$\downarrow$ & \textbf{Msg./Task}\,$\downarrow$ \\
\midrule
\multirow{3}{*}{Full conn. (reference)}
 & PI        & $0.839${\scriptsize$\pm0.058$} & $1202${\scriptsize$\pm47$}  & $57.6${\scriptsize$\pm9.0$}   & $0.0$ & $14.4${\scriptsize$\pm4.5$} & $5.08${\scriptsize$\pm0.84$} \\
 & CBBA      & $0.821${\scriptsize$\pm0.060$} & $1256${\scriptsize$\pm93$}  & $59.6${\scriptsize$\pm11.3$}  & $0.0$ & $25.1${\scriptsize$\pm6.0$} & $4.25${\scriptsize$\pm0.42$} \\
 & PI-MaxAss & $0.904${\scriptsize$\pm0.067$} & $1394${\scriptsize$\pm58$}  & $107.3${\scriptsize$\pm11.8$} & $0.0$ & $0.0$ & $11.31${\scriptsize$\pm1.47$} \\
\midrule
\multirow{3}{*}{Constrained (naive)}
 & PI-C        & $0.598${\scriptsize$\pm0.075$} & $1313${\scriptsize$\pm110$} & $76.2${\scriptsize$\pm40.7$}  & $0.0$ & $0.0$ & $4.40${\scriptsize$\pm0.68$} \\
 & CBBA-C      & $0.592${\scriptsize$\pm0.074$} & $1315${\scriptsize$\pm109$} & $81.6${\scriptsize$\pm43.4$}  & $0.0$ & $0.0$ & $3.66${\scriptsize$\pm0.50$} \\
 & PI-MaxAss-C & $0.251${\scriptsize$\pm0.040$} & $1294${\scriptsize$\pm120$} & $116.8${\scriptsize$\pm71.8$} & $0.0$ & $0.0$ & $21.04${\scriptsize$\pm3.44$} \\
\midrule
Constrained (geographic)
 & GEO & $0.721${\scriptsize$\pm0.069$} & $1252${\scriptsize$\pm104$} & $119.0${\scriptsize$\pm17.1$} & $77.3${\scriptsize$\pm19.7$} & $49.1${\scriptsize$\pm8.8$} & $49.21${\scriptsize$\pm11.96$} \\
\midrule
\multirow{2}{*}{Constrained (matched)}
 & PI-OM   & $0.730${\scriptsize$\pm0.045$} & $1179${\scriptsize$\pm119$} & $93.0${\scriptsize$\pm12.0$}  & $48.0${\scriptsize$\pm9.7$}  & $18.5${\scriptsize$\pm5.9$} & $72.06${\scriptsize$\pm11.61$} \\
 & CBBA-OM & $0.708${\scriptsize$\pm0.069$} & $1168${\scriptsize$\pm105$} & $97.3${\scriptsize$\pm15.1$}  & $55.4${\scriptsize$\pm14.0$} & $890${\scriptsize$\pm830$}  & $123.35${\scriptsize$\pm39.54$} \\
\midrule
Ours & \textbf{CC-OPI} & $0.789${\scriptsize$\pm0.054$} & $1294${\scriptsize$\pm92$} & $105.0${\scriptsize$\pm15.4$} & $44.7${\scriptsize$\pm12.0$} & $19.3${\scriptsize$\pm4.8$} & $54.74${\scriptsize$\pm8.93$} \\
\bottomrule
\end{tabular}
\end{table*}

\begin{table}[tbp]
\centering
\caption{Communication Accounting per Completed Task at $10$ UAVs / $30$ Tasks (Tight-Constraint Setting)}
\label{tab:comm}
\renewcommand{\arraystretch}{1.25}
\setlength{\tabcolsep}{4pt}
\footnotesize
\begin{tabular}{llrrrr}
\toprule
\textbf{Regime} & \textbf{Method} & \textbf{Bcast.} & \textbf{Deliv.} & \textbf{Sent (kB)} & \textbf{Recv.\ (kB)} \\
\midrule
\multirow{3}{*}{\shortstack[l]{Full conn.\\(reference)}}
 & PI          & $5.1$   & $45.7$  & $2.2$  & $19.8$ \\
 & CBBA        & $4.2$   & $38.2$  & $1.8$  & $16.6$ \\
 & PI-MaxAss   & $11.3$  & $101.8$ & $4.9$  & $44.2$ \\
\midrule
\multirow{3}{*}{\shortstack[l]{Constrained\\(naive)}}
 & PI-C        & $4.4$   & $0.0$   & $1.9$  & $0.0$ \\
 & CBBA-C      & $3.7$   & $0.0$   & $1.6$  & $0.0$ \\
 & PI-MaxAss-C & $21.0$  & $0.0$   & $9.1$  & $0.0$ \\
\midrule
\multirow{2}{*}{\shortstack[l]{Constrained\\(matched)}}
 & PI-OM       & $72.1$  & $35.1$  & $32.3$ & $16.0$ \\
 & CBBA-OM     & $123.3$ & $86.8$  & $55.4$ & $39.4$ \\
\midrule
Ours & \textbf{CC-OPI} & $54.7$ & $14.5$ & $24.4$ & $6.5$ \\
\bottomrule
\end{tabular}
\end{table}

\subsubsection{Default Setting}
The same qualitative pattern holds under the relaxed default setting: CC-OPI attains a TCR of about $0.97$--$0.98$, close to the full-connectivity reference, whereas PI-C and CBBA-C drop to roughly $0.71$--$0.78$ and PI-MaxAss-C falls to $0.22$ at the largest scale---the collapse is a structural consequence of intermittent connectivity, not an artifact of the tight-constraint setting.

\subsection{Ablation Study}
\label{subsec:ablation}

To isolate the contribution of each mechanism, we disable the spatial locality penalty, the state-locking mechanism, and the emergency pool one at a time, and compare each variant against the full model. All variants run on the radius-limited network over the default scales, since enabling full connectivity would mask the differences that these mechanisms address. A sunk-cost factor explored in early development was found to have a negligible effect and was removed from the algorithm; it is therefore not among the ablated mechanisms. Because the TCR is near saturation under this setting, each mechanism is assessed primarily through the metric it most directly affects rather than through the TCR alone. Table~\ref{tab:ablation} reports the percent change of each variant relative to the full model, averaged over the scales; the per-scale curves, omitted for space, show that these trends are consistent across scales.

\begin{table}[tbp]
\centering
\caption{Ablation: Percent Change Relative to the Full Model}
\label{tab:ablation}
\renewcommand{\arraystretch}{1.25}
\setlength{\tabcolsep}{5pt}
\begin{tabular}{lccc}
\toprule
\textbf{Metric} & \textbf{w/o Spatial} & \textbf{w/o StateLock} & \textbf{w/o Emergency} \\
\midrule
Task completion rate & $+0.1\%$  & $-0.7\%$ & $-3.6\%$ \\
Wasted distance      & $+26.5\%$ & $+4.5\%$ & $-7.9\%$ \\
Reassignment count   & $+37.6\%$ & $-2.4\%$ & $-12.1\%$ \\
Total distance       & $+12.0\%$ & $+2.8\%$ & $-15.6\%$ \\
Makespan             & $-7.2\%$  & $+4.7\%$ & $-9.2\%$ \\
\bottomrule
\end{tabular}
\end{table}

\subsubsection{Spatial Locality Penalty}
Removing the spatial penalty leaves the TCR essentially unchanged but inflates the coordination cost markedly: the wasted distance rises by $26.5\%$, the reassignment count by $37.6\%$, and the total flight distance by $12.0\%$ (Table~\ref{tab:ablation}). Without the penalty, disconnected UAVs bid greedily on the same remote, high-value tasks, so once subnets reconnect, consensus forces costly handovers and redundant travel. The spatial penalty thus delivers its intended effect---suppressing blind cross-regional chasing and the oscillation it induces---without sacrificing completion. The mild reduction in makespan when the penalty is removed indicates that unconstrained chasing can occasionally finish a few tasks sooner, but at a disproportionate cost in wasted travel.

\subsubsection{Emergency Pool}
The emergency pool exhibits a clear trade-off. Disabling it lowers the TCR by $3.6\%$ in relative terms, while simultaneously reducing the total distance ($-15.6\%$), wasted distance ($-7.9\%$), reassignments ($-12.1\%$), and makespan ($-9.2\%$). The pool spends additional flight effort to rescue tasks approaching their deadlines that the spatial penalty would otherwise cause UAVs to neglect; switching it off recovers those costs but forfeits the rescued completions. The mechanism should therefore be understood as trading travel for TCR under time pressure, not as a means of reducing distance. For the variants that retain the pool, the emergency rescue rate is high (about $0.92$--$0.96$ across scales), confirming that the flagged near-deadline tasks are in fact being recovered.

\subsubsection{State Lock}
\label{subsubsec:lock}
In the cross-scale aggregate, removing the state lock produces only small changes in the wasted distance ($+4.5\%$) and in the TCR ($-0.7\%$; the strict service-start accounting reveals that, without the lock, a handful of completions start after their deadline), which might suggest that the lock is inconsequential. Two considerations show otherwise. First, the lock is foremost a correctness and stability mechanism: it protects an in-progress or irreversible physical action from preemption and underpins the finite-time termination guarantee of Section~\ref{sec:theoretical}---without it, $\Omega_{lock}$ is empty and the eventual-commitment property (Lemma~\ref{lem:commit}) loses its premise---a role that an aggregate efficiency metric does not capture. Second, its efficiency benefit is masked by scale aggregation and becomes visible only when the connectivity regime is varied at a fixed scale.

\begin{figure}[tbp]
\centering
\includegraphics[width=0.95\columnwidth]{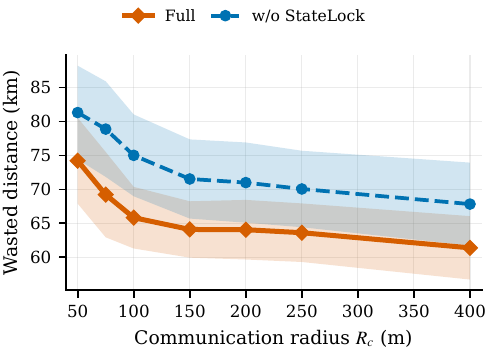}
\caption{State-lock ablation versus communication radius $R_c$ at a fixed scale.}
\label{fig:lock_ablation}
\end{figure}

To expose this, we sweep the communication radius $R_c$ at a fixed scale and compare the full model against the no-lock variant on shared seeds (Fig.~\ref{fig:lock_ablation}). With the lock enabled, the wasted distance is consistently $9\%$--$12\%$ lower across the entire range, a per-radius paired difference of $6$--$10$~km whose $95\%$ CI excludes zero at every radius; the makespan is directionally lower as well, although its per-radius paired CIs generally include zero. Crucially, the benefit is largest at small-to-medium radii ($R_c\approx75$--$100$~m) and decays as $R_c$ grows. This trend matches the mechanism's rationale: at smaller radii the network fragments more frequently, creating more opportunities for a UAV to be preempted mid-flight, so protecting its committed approach yields a larger saving; as connectivity improves, such mid-flight conflicts become rare and the lock's effect diminishes. The aggregate ablation therefore understates the lock's role, which the radius-resolved view makes explicit.

\subsubsection{Cumulative Ablation from the Online Backbone}
\label{subsubsec:backbone}
The preceding ablations disable one mechanism at a time. To connect this mechanism-level view with the paradigm-level attribution of Section~\ref{subsec:main_comparison}, we additionally enable the mechanisms cumulatively on top of the bare online backbone---the event-driven execution with every CC-OPI-specific score disabled ($\lambda=0$, no state lock, no emergency pool, $\beta=0$)---at the representative tight-constraint scale on the same $30$ seeds. The backbone retains the completed-set synchronization, which is correctness infrastructure for online execution rather than a scoring mechanism, and it is statistically indistinguishable from PI-OM (paired difference $+0.5$ percentage points, $95\%$ CI $[-0.7, +1.7]$), cross-validating the two implementations. Table~\ref{tab:backbone} reports the ladder, in which each row enables one further mechanism on top of the preceding row and $\Delta$TCR is the per-seed paired increment over that row with a bootstrap $95\%$ CI; the enabling order was fixed in advance and is not reordered post hoc.

\begin{table}[!t]
\centering
\caption{Cumulative (Backbone) Ablation at $10$ UAVs / $30$ Tasks (Tight-Constraint Setting)}
\label{tab:backbone}
\renewcommand{\arraystretch}{1.25}
\setlength{\tabcolsep}{2pt}
\footnotesize
\begin{tabular}{lcccc}
\toprule
\textbf{Variant} & \textbf{TCR}\,$\uparrow$ & \textbf{$\Delta$TCR (pt)} & \textbf{Wasted (km)} & \textbf{Msg./Task} \\
\midrule
Backbone            & $0.734$ & ---                          & $48.3$ & $109.9$ \\
$+$Spatial          & $0.759$ & $+2.5$\,$[+0.8, +4.0]$       & $44.5$ & $100.8$ \\
$+$Lock             & $0.771$ & $+1.2$\,$[+0.6, +1.9]$       & $44.1$ & $52.1$ \\
$+$Urgency          & $0.766$ & $-0.5$\,$[-1.2, +0.1]$       & $43.8$ & $54.1$ \\
$+$Emergency (Full) & $0.789$ & $+2.2$\,$[+0.8, +3.7]$       & $44.7$ & $54.7$ \\
\bottomrule
\end{tabular}
\end{table}

The ladder decomposes the mechanism gain. The spatial penalty and the emergency pool carry the largest completion improvements ($+2.5$ and $+2.2$ points, each with a CI excluding zero), and the spatial penalty simultaneously reduces the wasted distance and the reassignments, consistent with Section~\ref{subsec:ablation}. The state lock contributes a statistically resolved $+1.2$ points---part of which the strict service-start accounting uncovers, by no longer crediting the late completions that the unlocked variants accumulate---and, more strikingly, halves the communication cost: the messages per completed task drop from about $100$ to $52$ when the lock is enabled, because irrevocable ownership near completion stops the consensus vectors from oscillating and thereby suppresses the event-driven rebroadcasts---the same instability that inflates the message and reassignment counts of CBBA-OM in Section~\ref{subsec:main_comparison}. The urgency rung, in contrast, yields no measurable TCR gain on its own ($-0.5$ points, CI spanning zero, consistent with the $\beta$ sweep of Section~\ref{subsec:hyperparam}); we accordingly do not claim the urgency term as an independently validated contribution, but retain it as a design component whose intended effect is subsumed by the emergency layer. In total, the mechanisms add $+5.4$ points over the backbone ($95\%$ CI $[+3.4, +7.3]$), which, combined with the backbone--PI-OM equivalence, reproduces the $+5.9$-point gap between CC-OPI and PI-OM at this scale.

\subsection{Sensitivity to Hyperparameters}
\label{subsec:hyperparam}

CC-OPI exposes a small set of hyperparameters that were fixed to a single configuration across every experiment above and never tuned per scenario. To verify that the preceding conclusions are not an artifact of this particular choice, we vary the three free strength and threshold parameters one at a time at the representative tight-constraint scale---the spatial-penalty strength $\lambda$, the urgency-penalty weight $\beta$, and the emergency-pool margin $T_{emergency}$---holding the remaining parameters, which are fixed by physical or temporal reasoning, at their defaults. The TCR is insensitive to $\lambda$ and $\beta$ within one SEM: the benefit of the spatial penalty appears in the coordination cost (disabling it raises the wasted distance by $11\%$ and the reassignment count by $18\%$) and plateaus for $\lambda\gtrsim5$, and the weak effect of $\beta$ is consistent with the unresolved urgency rung of Section~\ref{subsubsec:backbone}. The emergency margin $T_{emergency}$ is the one genuinely sensitive parameter: an overly large margin flags too many tasks as emergencies too early, so the spatial-penalty exemption ($\mu=1$) erodes the regionalized operation, and the TCR degrades gradually from about $0.80$ to $0.74$ as the margin grows from the default $700$~s to $1500$~s, the default sitting inside a wide safe plateau. The attractiveness bonus $\mathcal{A}_{bonus}$ acts as a soft infinity whose only requirement is to be large enough to favor a flagged task in the inclusion ranking: a sweep over $\mathcal{A}_{bonus}\in[500, 20000]$ produces identical results at every value, and removing the bonus entirely lowers the TCR and the emergency rescue rate by less than one percentage point each (within one SEM)---the cross-regional rescues are carried by the spatial-penalty exemption ($\mu=1$) rather than by the bonus. The full protocol and per-parameter figures are given in the supplementary material (Fig.~S3). Because the parameter values were fixed once and held fixed across all experiments, the gains reported in Sections~\ref{subsec:main_comparison}--\ref{subsec:robustness} are not the product of per-scenario tuning.

\subsection{Sensitivity to the Communication Radius}
\label{subsec:radius}

The premise of this work is a limited communication radius, so we now examine how performance varies with $R_c$ as the explicit independent variable. We fix the tight-constraint scenario at $10$ UAVs and $30$ tasks and sweep $R_c$ from $75$~m to $15{,}000$~m, the latter exceeding the map diagonal and thus approximating full connectivity. To isolate the effect of the radius, terrain occlusion is disabled for this experiment; line-of-sight blockage is treated separately as a robustness factor. CC-OPI runs on the radius-limited network; the static methods run as constrained baselines (PI-C, CBBA-C, PI-MaxAss-C), and their full-connectivity counterparts are drawn as horizontal reference lines. Fig.~\ref{fig:radius} reports the TCR, wasted distance, and communication overhead against $R_c$.

\emph{CC-OPI degrades gracefully over the swept range, whereas the static baselines collapse.} As $R_c$ decreases from near-full connectivity to $75$~m, the TCR of CC-OPI declines gently from about $0.88$ to $0.79$, a drop of roughly ten percentage points. Over the same range, the constrained baselines fall steeply: PI-C and CBBA-C drop from about $0.84$ and $0.82$ to about $0.60$, and PI-MaxAss-C from $0.90$ to $0.25$. At the representative deployment radius $R_c=250$~m, the margins---about $20$ percentage points over PI-C and CBBA-C and more than $50$ over PI-MaxAss-C---agree with Table~\ref{tab:hard_summary}. Notably, the baselines are insensitive to $R_c$ below roughly $500$~m: at these radii the initial communication graph is typically empty, so the static methods reduce to independent greedy allocation regardless of the exact value of $R_c$.

\emph{The large-radius limit validates the comparison and locates the crossover.} As $R_c$ grows, each constrained baseline converges to its own full-connectivity reference (PI-C to $0.839$, CBBA-C to $0.821$, and PI-MaxAss-C to $0.904$, matching Table~\ref{tab:hard_summary}), confirming that the degradation at small $R_c$ is attributable to the loss of connectivity alone. Within the swept range, CC-OPI leads every evaluated baseline for all $R_c\le4000$~m. The only method that overtakes it is PI-MaxAss-C, and only once the network is nearly complete ($R_c\gtrsim8000$~m), where its two-stage global re-optimization regains the conditions for which it was designed. The resulting crossover lies at a radius of several kilometers---more than an order of magnitude above any radius representative of a post-disaster deployment---so within the regime of interest CC-OPI is the only method evaluated that sustains a high TCR. We further note that, at the full-connectivity end of the sweep, CC-OPI ($0.883$) even exceeds the static PI ($0.839$) and CBBA ($0.821$), since its continuous online execution does not freeze after an initial allocation; we report this observation without claiming it as a design objective.

\emph{Coordination cost scales with connectivity.} The secondary metrics expose the cost structure of the online paradigm (Fig.~\ref{fig:radius}). As $R_c$ increases, CC-OPI's wasted distance falls sharply, from about $46$~km at $R_c=75$~m to under $1$~km near full connectivity, because better connectivity resolves conflicts earlier and avoids redundant travel. Conversely, its communication overhead rises monotonically, from about $52$ to $325$ messages per completed task, as a larger radius admits more neighbors and thus more frequent broadcasts. CC-OPI therefore exchanges wasted travel for messages as connectivity improves; at the representative radius it operates in the regime of higher wasted distance and lower message count.

\begin{figure*}[!t]
\centering
\includegraphics[width=0.9\textwidth]{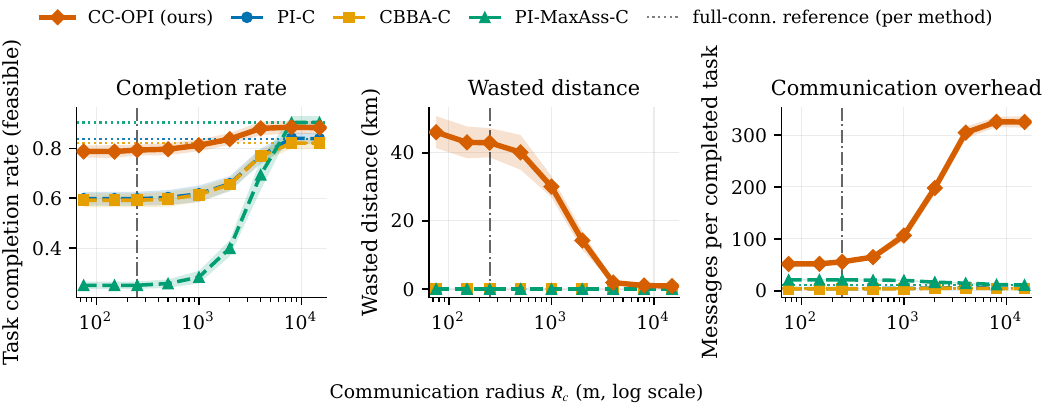}
\caption{Sensitivity to the communication radius $R_c$ at a fixed scale.}
\label{fig:radius}
\end{figure*}

\subsection{Robustness to Packet Loss, Terrain Occlusion, and Dynamic Task Arrival}
\label{subsec:robustness}

\subsubsection{Link-Level Packet Loss}
\label{subsubsec:linkloss}
The evaluation protocol so far delivers every in-range message (Table~\ref{tab:assumptions}). We now relax this idealization: each delivery of a broadcast to one in-range receiver is dropped independently with probability $p \in \{0.1, 0.3, 0.5\}$, while broadcasts remain counted as sent, preserving the message-metric semantics. The study uses the representative tight-constraint scale ($10$ UAVs, $30$ tasks) at $R_c=250$~m on the same $30$ seeds and compares CC-OPI with the matched online baselines PI-OM and CBBA-OM under the same loss probability. The frozen static baselines are omitted: their single consensus round operates on the typically empty initial graph at this radius (Section~\ref{subsec:radius}), so message loss leaves them essentially unchanged. Fig.~\ref{fig:linkloss} reports the results.

\begin{figure}[tbp]
\centering
\includegraphics[width=0.92\columnwidth]{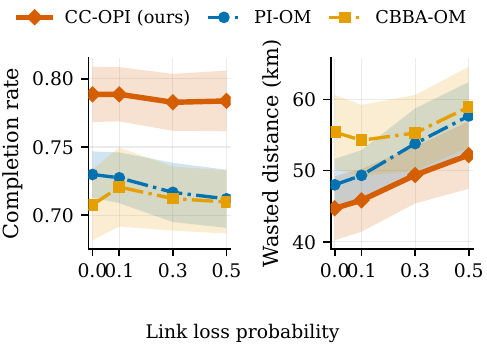}
\caption{Robustness to link-level packet loss at a fixed scale.}
\label{fig:linkloss}
\end{figure}

CC-OPI is nearly indifferent to the loss rate. Its TCR remains at $0.78$--$0.79$ across the sweep---the paired difference against the loss-free runs is $-0.6$ percentage points ($95\%$ CI $[-1.7, +0.5]$) at $30\%$ loss and $-0.5$ ($[-1.6, +0.7]$) at $50\%$, indistinguishable from zero even when half of all deliveries are dropped---and the emergency rescue rate stays at about $0.74$. The cost surfaces in redundant travel: the wasted distance grows by $4.7$~km ($[+1.2, +8.3]$) at $30\%$ and by $7.5$~km ($[+4.2, +11.1]$) at $50\%$ loss, about $17\%$ above the loss-free level, because lost messages defer conflict resolution until the UAVs meet again or converge physically. The messages per completed task even decline slightly (from about $55$ to $51$): fewer received messages trigger fewer state changes and hence fewer event-driven rebroadcasts.

The margin over the matched baselines widens rather than erodes. PI-OM degrades significantly ($-1.8$ points, $[-3.0, -0.7]$, at $50\%$ loss), so the paired advantage of CC-OPI grows from $+5.9$ points ($[+3.9, +7.6]$) without loss to $+7.2$ ($[+4.9, +9.3]$) at $50\%$, and remains between $+6.8$ and $+8.1$ points over CBBA-OM throughout. Notably, CBBA-OM improves slightly under mild loss ($+1.3$ points, $[+0.1, +2.6]$, at $10\%$), its median reassignment count falling from $660$ to $113$: random loss happens to interrupt the bid oscillation that commitment protection suppresses by design (Section~\ref{subsubsec:backbone})---a symptom of the pathology rather than a virtue, since its TCR still trails CC-OPI by about seven points.

This tolerance is structural: no CC-OPI mechanism assumes that a message arrives. The monotone completed set recovers a lost completion fact at any later contact, the timestamped consensus rules reconcile late-arriving state correctly, the periodic heartbeat acts as a natural retransmission, and the state lock ties physical commitment to position rather than to message arrival. A lost delivery therefore translates into delay---and ultimately into the moderate additional travel quantified above---rather than into an inconsistent allocation.

Burst-correlated loss and delivery delay do not alter this picture. Replacing the independent drops with a two-state Gilbert--Elliott gate per link---parameterized so that the stationary loss equals each tested level while outages arrive in bursts of five seconds on average---leaves the TCR of CC-OPI statistically unchanged at every level (paired difference at $50\%$ mean loss: $-0.2$ points, $95\%$ CI $[-1.3, +0.8]$), while PI-OM again degrades ($-1.2$ $[-2.0, -0.4]$) and the margin over PI-OM widens to $+6.8$ $[+4.9, +8.6]$. Subjecting every surviving delivery to a uniform random delay of up to ten seconds---under which messages arrive out of order---likewise leaves the TCR of CC-OPI unchanged ($-0.1$ $[-1.2, +1.0]$ at the harshest level) at the cost of $5.4$~km $[+1.3, +9.4]$ of additional redundant travel: the same structural properties reconcile late, out-of-order state correctly. The full protocols and per-level results are given in the supplementary material. Interference, contention, and channel capacity remain unmodeled (Table~\ref{tab:assumptions}).

\subsubsection{Terrain Occlusion}
Terrain can additionally block line-of-sight links: a communication edge exists only if the two UAVs are within range \emph{and} no obstacle intersects the segment between them. In an isolated communication-occlusion test---conducted at a radius large enough that the obstacle-free network is essentially fully connected, since at $R_c=250$~m the distance cutoff alone already disconnects almost every pair---sweeping from $0$ to $10$ large random obstacles lowers the TCR of CC-OPI only modestly, from about $0.88$ to $0.87$ (paired difference $-1.6$ percentage points, $95\%$ CI $[-2.7, -0.5]$), while the wasted distance rises roughly fivefold, because conflicts that early consensus would resolve are instead discovered later through physical convergence. Occlusion blocks communication only; physical motion follows straight lines without obstacle avoidance. The full protocol and per-metric results are given in the supplementary material (Fig.~S2).

\subsubsection{Dynamic Task Arrival}
\label{subsec:dynamic}
Search and rescue missions must also cope with tasks that emerge after deployment, such as newly discovered survivors. The event-driven design of CC-OPI accommodates this without network-wide suspension: a newly injected task is announced as a new-task wake-up event and processed as an ordinary local trigger, with the same priority as a topology change or a completion. We adopt a global-discovery model, in which the \emph{existence} of an injected task is announced to all UAVs upon arrival, while its \emph{allocation} remains subject to the communication constraints. We do not claim dynamic arrival as a core capability; rather, we use it to assess the robustness of the online paradigm to runtime changes.

We fix the tight-constraint setting at $10$ UAVs and $30$ base tasks at the representative radius $R_c=250$~m with clear terrain, and inject an increasing number of additional tasks during the mission at a fixed cadence (one task every $120$~s); each injected task is given a deadline relative to its arrival time. Fig.~\ref{fig:dynamic} reports the overall TCR and the injected-task completion rate as the number of injected tasks grows from $0$ to $10$ (a $33\%$ increase in total workload); the associated costs are quoted below.

The overall TCR remains essentially flat, between $0.75$ and $0.79$, as the injected workload grows. Because the injected tasks enter both the numerator and the denominator of this metric, we also examine the base tasks separately: their completion rate declines by $0.7$ percentage points at two injections (paired against the injection-free runs, $95\%$ CI $[-1.3, -0.2]$) and settles at $-1.8$ points ($[-2.7, -1.0]$) from six injections onward---indeed, the per-seed base-task outcomes at six and ten injections are identical, a direct between-level difference of exactly zero. Runtime arrivals therefore displace the existing allocation only by a small amount over the tested $0$--$10$ injection levels---the expected price of admitting new tasks under a shared capacity---rather than destabilizing it. For the injected tasks, since each run contributes only a few, we report the completion rate pooled across seeds per injection level, with a $95\%$ confidence interval from a run-level cluster bootstrap that resamples whole runs and thus preserves within-run correlation (a pooled Wilson interval \cite{brown2001interval} gives nearly identical bounds). The pooled rate rises from $0.47$ (CI $[0.33,0.60]$ over $60$ injected tasks) at two injections to $0.77$ (CI $[0.73,0.81]$ over $300$ tasks) at ten; the direct between-level differences are positive with intervals excluding zero ($+18.9$ points $[+8.3, +28.9]$ from two to six injections and $+11.4$ $[+8.2, +14.7]$ from six to ten, paired seed-cluster bootstrap), so the increase is supported by a direct test. The upward trend reflects the injection schedule as much as the algorithm: at the fixed $120$~s cadence, a higher injection count also extends the arrival window, so more tasks arrive after the initial workload has largely cleared and meet more available capacity. We therefore do not read the trend as evidence that higher intensities are inherently easier to absorb. The cost scales with the added workload---the total distance, makespan, and reassignment count all increase as more tasks are served---whereas the communication overhead per completed task stays nearly constant (about $54$--$58$ messages), indicating that the per-task coordination effort does not inflate with the injection load. Overall, these results indicate that, under the tested fixed-cadence schedule, the event-driven paradigm handles runtime task arrival gracefully, in contrast to the static paradigm, which would require a network-wide replanning pause that is unachievable under partitioned connectivity.

\begin{figure}[tbp]
\centering
\includegraphics[width=0.92\columnwidth]{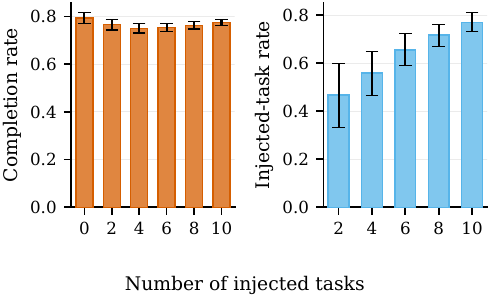}
\caption{CC-OPI under dynamic task injection at a fixed scale.}
\label{fig:dynamic}
\end{figure}

\section{Conclusion and Future Work}
\label{sec:conclusion}

This paper addressed distributed task allocation under restricted communication, where a limited communication radius fragments the UAV swarm into transient information islands and invalidates the static ``allocate-then-execute'' paradigm. We proposed CC-OPI, an event-driven online algorithm that interleaves task negotiation with physical execution rather than separating the two. CC-OPI rests on three components: an asynchronous online framework that replans only at discrete physical and topological events; a pair of cost-evaluation metrics tailored to dynamic topologies---a dynamic IPI with a spatial locality penalty and an effective RPI with deadline-aware urgency---complemented by a non-preemptive state lock that protects each UAV's ongoing physical action; and a decentralized fault-tolerance layer that combines version-based state synchronization with a global-time-driven emergency pool. We further established that CC-OPI terminates in finite time, free of stale-completion deadlock and of unbounded reassignment within the mission horizon, at a per-wake-up cost that is polynomial in the problem size (Section~\ref{sec:theoretical}).

Beyond these guarantees, extensive simulations characterize the practical performance of CC-OPI. At the representative communication radius of $250$~m, CC-OPI sustains a TCR of about $0.80$ that stays essentially flat across problem scales, exceeding the naively transferred static baselines by roughly $18$ to $22$ percentage points while remaining comparable to CBBA and within roughly six percentage points of PI, both evaluated under full connectivity. A matched online comparison decomposes this margin: executing the unmodified PI rules in the same event-driven regime recovers about $12$ points over the frozen transfer, and CC-OPI's mechanisms add a further $7$ points ($95\%$ CI $[+6.1, +7.8]$, replicated on fresh hold-out seeds) while incurring less wasted travel and fewer messages per completed task than the matched baselines. As the radius shrinks, CC-OPI degrades gracefully where the static baselines collapse, leading every constrained baseline evaluated until the network is nearly complete; it also absorbs the tested link-level packet loss, terrain occlusion, and runtime task arrival with only a modest loss in TCR. These gains are not free. Because CC-OPI coordinates continuously throughout execution, it incurs a higher communication volume---roughly an order of magnitude more application-layer broadcasts and transmitted bytes per completed task than the static methods, although in successful deliveries it is cheaper than the full-connectivity consensus (Section~\ref{sec:simulation})---together with some redundant travel from online reassignment. CC-OPI therefore favors robustness under intermittent connectivity over optimality under idealized assumptions, and this cost, together with the simplified physical model, delineates the boundaries of the present study.

These boundaries point to two complementary avenues for future work. The first targets the coordination cost of the online paradigm. As reported in Section~\ref{sec:simulation}, CC-OPI trades a higher communication volume and some redundant travel for its robustness, and narrowing this gap---for example through adaptive broadcast suppression, message compression, and leaner conflict resolution---is a direct line of improvement. Adapting the algorithm's parameters to the operating conditions offers a further lever in the same direction: although a single fixed configuration is robust across all evaluated scenarios (Section~\ref{sec:simulation}), the state-lock ablation also shows that the benefit of individual mechanisms varies with the connectivity regime. A policy that adjusts the spatial, urgency, and locking parameters to the observed connectivity---learned, for instance, through multi-agent reinforcement learning---could therefore reduce the coordination overhead below what any fixed configuration achieves.

The second avenue concerns the physical fidelity of the model. The present kinematics assume unlimited endurance and constant-speed straight-line motion: the termination horizon is bounded by task deadlines alone, with no energy constraint, and terrain occlusion affects only connectivity rather than the flight path. Incorporating energy budgets together with return-to-base recharging would couple the battery state with the allocation decision, so that a detour consumes a finite resource rather than time alone. Integrating obstacle-aware path planning would in turn allow terrain to constrain the physical trajectory as well, unifying communication blockage and motion blockage within a single model. Resilience to UAV failure belongs to the same agenda: although the fault-tolerance layer is motivated in part by the loss of task owners, mid-mission UAV failures are not exercised in our experiments, and evaluating how gracefully the swarm absorbs them remains open. We regard these extensions as necessary steps toward deploying the online paradigm under field conditions.

\section*{Data and Code Availability}
The complete simulation platform, all experiment runners, the raw per-run result data, and the plotting and statistical-analysis scripts that produce every figure, table, and confidence interval in this paper are publicly available at \url{https://github.com/bdathe-lb/CC-OPI-Exp} (release \texttt{v1.0-tmc}, MIT license). The repository documents the pinned software environment (Python~3.12), the seed protocol (seeds $42$--$71$, shared across all experiments), the reachable-task exclusion criterion, the runtime of each experiment suite, and the mapping from each raw data file to its figure or table; the committed raw data allow the analysis and plots to be reproduced without re-running the simulations. The matched online wrappers PI-OM and CBBA-OM are released in the same repository as first-class reproducible baselines.

\bibliographystyle{IEEEtran}
\bibliography{refs}

\vfill

\end{document}


\title{Supplementary Material for\\ ``CC-OPI: Online Distributed Task Allocation for UAV Swarms under Communication Constraints''}

\author{Liu Biao and Zhang Tong}

\maketitle

\noindent This document provides the complete proofs of Corollary~1, Lemma~3, Theorem~1, and Corollary~2 of the main paper, whose statements are repeated here for convenience, together with the supplementary trade-off figure (Fig.~\ref{fig:tradeoff_supp}), the full terrain-occlusion robustness experiment (Fig.~\ref{fig:occlusion_supp}), the full hyperparameter-sensitivity study (Fig.~\ref{fig:hyperparam}), the communication-graph statistics of the simulated mobility process (Table~\ref{tab:graph_stats}), the capability mapping of distributed allocation methods (Table~\ref{tab:capability}), the hold-out seed validation of the headline comparisons, the burst-loss and delivery-delay robustness experiments (Table~\ref{tab:linkstress}), and the scaling study (Table~\ref{tab:scaling}), all referenced in Sections~II and~VI of the main paper. All notation, definitions, and assumptions (A1)--(A5) follow Section~V of the main paper; equation, lemma, corollary, and remark numbers below refer to the main paper.

\begin{corollary}[Deadlock-freedom within a component]
Within any connected component whose members remain connected until the completed-set synchronization finishes, no completed task is re-included by its members; hence no coordination deadlock arises from stale completion information inside such a component.
\end{corollary}

\begin{IEEEproof}
Once a task is completed, its identity is added to the completed set of some member; by the version-based synchronization of Section~IV of the main paper, it propagates to every other member that remains in the same component within a finite number of synchronization rounds and, by Lemma~2, is never dropped thereafter. A member that still holds the task evicts it from its path upon synchronization and permanently excludes it from subsequent inclusion. Hence no member re-includes a completed task, and the component cannot deadlock on stale completion information. If the component splits before the synchronization finishes, the completion fact is not lost: it continues to spread along later opportunistic contacts (the gossip semantics of Section~IV of the main paper), reaching any UAV to which a temporal contact path exists. Stale information held across disconnected components is resolved by the consensus rules once the components merge, as noted in Remark~1 of the main paper.
\end{IEEEproof}

\begin{lemma}[Eventual commitment]
Let $u_i$ move toward its primary target $\tau$, the head of its path. Once the remaining distance falls below $d_{lock}$, $\tau \in \Omega_{lock}$ and $RPI_{eff}(\tau) = -\infty$; thereafter $\tau$ is neither preempted through consensus nor displaced from the head of the path. The task $\tau$ then leaves the path of $u_i$ within $d_{lock}/v_{min} + s_\tau$, either served by $u_i$ or, if another UAV is already serving it, dropped by $u_i$ upon arrival. Moreover, if $\tau$ is deadline-feasible at the instant it is locked, this resolution is an on-time completion by exactly one UAV.
\end{lemma}

\begin{IEEEproof}
When $\tau$ is locked, its effective score $RPI_{eff}(\tau) = -\infty$ is the strongest possible retention bid, so no neighbor can win $\tau$ away from $u_i$ under the consensus rules; the removal phase reclaims a locked task rather than releasing it, and the non-preemption constraint forbids the inclusion phase from inserting any task ahead of a locked head. Hence $\tau$ remains the head of the path of $u_i$ until physically resolved. The remaining distance is below $d_{lock}$ and the speed is at least $v_{min}$, so $u_i$ reaches $\tau$ within $d_{lock}/v_{min}$. Two UAVs in the same component may both lock $\tau$, since locking is a local decision; in that case the consensus tie is broken by the order in~(A5), yet the loser still retains $\tau$ locally and continues toward it. Upon arrival, $u_i$ either finds $\tau$ unserved and serves it within $s_\tau$ (the service being non-interruptible), or finds it already served and drops it; in either case $\tau$ leaves the path of $u_i$ within $d_{lock}/v_{min} + s_\tau$. Finally, the estimated arrival time $A_{i,k} = t + dist_{rem}/v_i$ is invariant under direct flight at constant speed, so a task deadline-feasible at the instant of locking remains feasible until served, and its resolution is an on-time completion by exactly one UAV. A head already deadline-infeasible when locked is still resolved within the same bound, but not as an on-time completion: the UAV proceeds to it and, upon arrival, either drops it (if already served) or serves it late---a physical completion whose service starts after the deadline and is therefore excluded from the completion metric (Section~VI of the main paper). Either outcome removes the head from the path, which suffices for termination.
\end{IEEEproof}

\begin{theorem}[Finite-time termination]
Under Assumption~1, CC-OPI reaches a terminal state within finite physical time and performs finitely many operations. Recalling the map diagonal $D_{map}$, the termination time satisfies
\begin{equation}
\label{eq:tstop_supp}
    T_{stop} \le T_{max} := \max_k D_k + \Bigl(\max_i C_i\Bigr)\!\left(\frac{D_{map}}{v_{min}} + \max_k s_k\right).
\end{equation}
\end{theorem}

\begin{IEEEproof}
\emph{(i) Bounded horizon.} By (A3), once $t > \max_k D_k$, every task is resolved: it is either already completed or expired, so no live task remains. Beyond this time no task is deadline-feasible, so the inclusion phase admits no new task and every path can only shrink. Each UAV then drains its remaining path of length at most $\max_i C_i$ by physically reaching its head task in turn; each such step takes at most $D_{map}/v_{min}$ of travel, since any leg is at most $D_{map}$ long and the speed is at least $v_{min}$, plus at most $\max_k s_k$ of non-interruptible service (Lemma~3), after which the task leaves the path. Once every path is empty and no UAV is serving, the terminal state is reached, which yields~\eqref{eq:tstop_supp}.

\emph{(ii) Finite computation.} By (A4), physical time advances in steps $\Delta t > 0$, so the number of frames is at most $T_{max}/\Delta t$. In each frame, every awakened UAV invokes the solver once, which is bounded by Lemma~1. The total number of operations is therefore finite.

\emph{(iii) Productive progress.} Beyond termination, the progress is productive. The task set is fixed and resolution is permanent, so $\Phi(t)$ is non-increasing; Corollary~1 with Lemma~2 precludes stale-completion deadlock within a component; and Lemma~3 ensures that a commitment reaching the $d_{lock}$ neighborhood is resolved within a bounded time---by an on-time completion whenever the task is still feasible when locked---rather than reversed. Hence the horizon is spent productively reducing $\Phi(t)$ rather than being lost to idling, deadlock, or oscillation.
\end{IEEEproof}

\begin{corollary}[Deadlock- and livelock-freedom]
Within the execution horizon of Theorem~1, CC-OPI admits neither coordination deadlock arising from stale completion information nor an unbounded reassignment sequence.
\end{corollary}

\begin{IEEEproof}
Deadlock-freedom within a component is Corollary~1, and stale information across components is resolved upon merging (Remark~1 of the main paper). For livelock-freedom, parts (i)--(ii) of Theorem~1 bound the number of wake-ups over $[0, T_{max}]$ by $N \lceil T_{max}/\Delta t \rceil$, and each wake-up modifies a path by a bounded amount (Lemma~1), while a locked commitment is irreversible (Lemma~3). The emergency bonus operates within this finite wake-up budget and hence introduces no unbounded re-attraction. The total number of reassignments is therefore finite.
\end{IEEEproof}

\section*{Supplementary Figure}

Fig.~\ref{fig:tradeoff_supp} places every method of the main performance comparison (Section~VI of the main paper) in a single completion-versus-cost plane at the representative scale.

\begin{figure}[!t]
\centering
\includegraphics[width=\columnwidth]{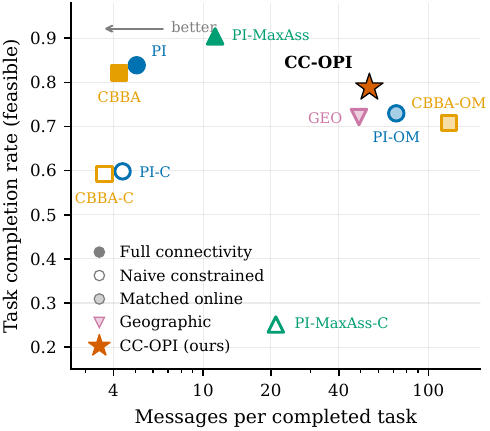}
\caption{Robustness-versus-cost trade-off at a representative scale ($10$ UAVs, $30$ tasks, tight-constraint setting).}
\label{fig:tradeoff_supp}
\end{figure}

\section*{Supplementary Experiment: Terrain Occlusion}

This section reports the full protocol and results of the terrain-occlusion robustness test summarized in Section~VI-F of the main paper.

Beyond the communication radius, terrain can block line-of-sight links: a communication edge exists only if the two UAVs are within range \emph{and} no obstacle intersects the segment between them. We study how CC-OPI tolerates such occlusion. Because occlusion only removes a link that is already within range, its effect is masked at the representative radius $R_c=250$~m, where the distance cutoff alone already disconnects almost every pair. To isolate the occlusion factor, we therefore set the radius large enough that the obstacle-free network is essentially fully connected, and then sweep the occlusion intensity by introducing an increasing number of large, randomly placed circular obstacles (radius $500$--$1200$~m) at the fixed scale of $10$ UAVs and $30$ tasks under the tight-constraint setting, averaging over the same $30$ seeds as the main experiments. We emphasize that occlusion blocks communication only: physical motion follows straight lines without obstacle avoidance, so tasks never become physically unreachable and only coordination is affected. The comparison is internal to CC-OPI, isolating the effect of occlusion rather than comparing algorithms. Fig.~\ref{fig:occlusion_supp} reports the results.

As the number of obstacles increases from $0$ to $10$, the task completion rate of CC-OPI declines only modestly, from about $0.88$ to $0.87$ (paired difference at ten obstacles: $-1.6$ percentage points, $95\%$ CI $[-2.7, -0.5]$---small but statistically resolved). The coordination cost, however, is more visible: the wasted distance rises consistently, from about $0.9$~km to $4.8$~km---a roughly fivefold increase---while the communication overhead falls from about $325$ to $251$ messages per completed task. This pattern admits a coherent explanation. Occlusion removes communication links, so conflicting claims that would otherwise be resolved through early consensus are instead discovered later, once UAVs converge in physical space, which incurs additional redundant travel; at the same time, fewer reachable neighbors mean fewer broadcasts. The makespan is essentially unaffected. In this isolated communication-occlusion test, CC-OPI thus absorbs line-of-sight blockage primarily as a moderate increase in redundant travel rather than as failed tasks.

\begin{figure}[!t]
\centering
\includegraphics[width=\columnwidth]{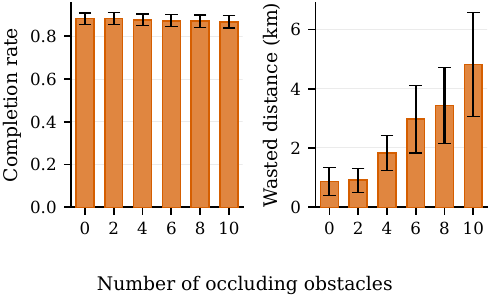}
\caption{Robustness of CC-OPI to terrain occlusion at a fixed scale ($10$ UAVs, $30$ tasks, tight-constraint setting).}
\label{fig:occlusion_supp}
\end{figure}

\section*{Supplementary Experiment: Hyperparameter Sensitivity}

This section reports the full protocol and per-parameter results of the hyperparameter-sensitivity study summarized in Section~VI-D of the main paper.

CC-OPI exposes a small set of hyperparameters that were fixed to a single configuration across every experiment of the main paper and never tuned per scenario. To verify that the reported conclusions are not an artifact of this particular choice, we vary the three free strength and threshold parameters one at a time, holding the others at their defaults: the spatial-penalty strength $\lambda$, the urgency-penalty weight $\beta$, and the emergency-pool margin $T_{emergency}$. The remaining parameters are fixed by physical or temporal reasoning rather than swept---the spatial exponent $\gamma=1$ (a deliberately linear penalty), the urgency exponent $\kappa$ and onset $t_{safe}$, and the lock distance $d_{lock}=50$~m (an irreversibility radius of about $1.7$~s and $1.0$~s of flight at the two cruising speeds). The attractiveness bonus $\mathcal{A}_{bonus}$ acts as a soft infinity whose only requirement is to be large enough to favor a flagged task in the inclusion ranking: a separate sweep over $\mathcal{A}_{bonus}\in[500, 20000]$ produces identical results at every value, and removing the bonus entirely ($\mathcal{A}_{bonus}=0$) lowers the TCR and the emergency rescue rate by less than one percentage point each---the cross-regional rescues are carried by the spatial-penalty exemption ($\mu=1$) rather than by the bonus. We use the tight-constraint setting at $10$ UAVs and $30$ tasks, the representative radius $R_c=250$~m, and clear terrain, averaging over the same $30$ seeds; shaded bands denote $95\%$ confidence intervals of the mean, and the standard error of the mean (SEM) of the TCR is about one percentage point. Fig.~\ref{fig:hyperparam} reports, for each parameter, the metric that carries its conclusion.

\emph{The TCR is insensitive to $\lambda$, whose benefit appears in the coordination cost and saturates quickly.} As $\lambda$ varies over $[0,50]$, the TCR remains near $0.80$ within a single SEM, so the spatial penalty does not trade away completions. Its effect is instead visible in the efficiency metrics (Fig.~\ref{fig:hyperparam_lambda}): disabling it ($\lambda=0$, which coincides with the ``w/o Spatial'' ablation variant of Section~VI-C of the main paper) raises the wasted distance from about $43$ to $48$~km ($+11\%$) and the reassignment count from about $19$ to $23$ ($+18\%$), consistent with the ablation. Both metrics flatten for $\lambda\gtrsim5$ and stay essentially constant through $\lambda=50$, so the default $\lambda=10$ lies on a broad plateau rather than at a sharp optimum.

\emph{The urgency weight $\beta$ has only a weak effect.} Across $\beta\in[0,2000]$ the TCR varies by about one percentage point---within one SEM (Fig.~\ref{fig:hyperparam_beta})---while the wasted distance grows only mildly with $\beta$. The default $\beta=500$ therefore sits in a stable region; we do not claim it to be optimal, and the insensitivity itself indicates that the reported performance does not hinge on this value.

\emph{The emergency margin $T_{emergency}$ is the one genuinely sensitive parameter, and the default lies within its safe region.} For $T_{emergency}\le700$~s the TCR stays near $0.80$, but it degrades to about $0.78$ at $1000$~s and $0.74$ at $1500$~s (Fig.~\ref{fig:hyperparam_temrg}), a drop of roughly five SEM. The trend follows the mechanism: an overly large margin flags too many tasks as emergencies too early, so the spatial-penalty exemption ($\mu=1$) applies to a large fraction of tasks and erodes the regionalized operation that suppresses blind cross-regional chasing; the emergency rescue rate falls in step. The accompanying decrease in total distance and makespan does not reflect greater efficiency but fewer completed tasks. The default $T_{emergency}=700$~s sits just inside the safe plateau; the marginally higher rate at $500$~s is within statistical error and we do not read it as a better setting.

Taken together, these sweeps indicate that the reported conclusions are robust to CC-OPI's hyperparameters: the TCR is insensitive to $\lambda$ and $\beta$ within statistical error, and $T_{emergency}$ admits a wide safe region around the default and degrades gradually even when grossly mis-specified.

\begin{figure*}[!t]
\centering
\subfloat[Spatial-penalty strength $\lambda$ (default $10$)]{\includegraphics[width=0.30\textwidth]{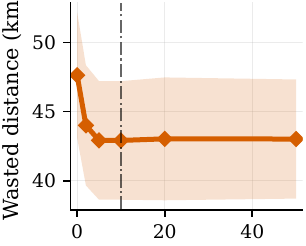}\label{fig:hyperparam_lambda}}
\hfil
\subfloat[Urgency-penalty weight $\beta$ (default $500$)]{\includegraphics[width=0.30\textwidth]{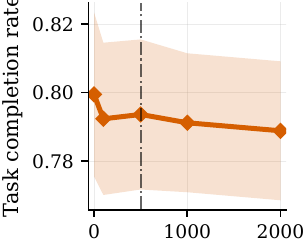}\label{fig:hyperparam_beta}}
\hfil
\subfloat[Emergency-pool margin $T_{emergency}$ (default $700$~s)]{\includegraphics[width=0.30\textwidth]{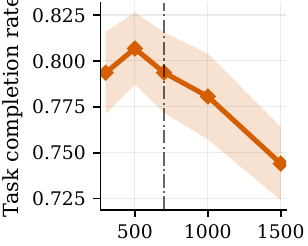}\label{fig:hyperparam_temrg}}
\caption{Hyperparameter sensitivity of CC-OPI at a fixed scale ($10$ UAVs, $30$ tasks, tight-constraint setting).}
\label{fig:hyperparam}
\end{figure*}

\section*{Supplementary Table: Communication-Graph Statistics}

Table~\ref{tab:graph_stats} characterizes, in graph terms, what each communication radius means in the simulated mobility process (Section~VI-A of the main paper), so that the results can be mapped to other environments. All quantities are averaged over the $30$ seeds of CC-OPI runs at the representative scale ($10$ UAVs, $30$ tasks, tight-constraint setting, clear terrain, matching the radius sweep of Section~VI-E); the statistics depend only on the UAV trajectories. At the representative setting with the default terrain ($R_c=250$~m), the values are essentially identical to the clear-terrain row (time-averaged mean degree $0.16$, largest-component fraction $0.15$, about $15$ contacts of median duration $12$~s per mission).

\begin{table}[!t]
\centering
\caption{Communication-graph statistics per radius ($10$ UAVs, $30$ tasks, clear terrain, means over $30$ seeds).}
\label{tab:graph_stats}
\renewcommand{\arraystretch}{1.2}
\setlength{\tabcolsep}{3.5pt}
\footnotesize
\begin{tabular}{rcccccc}
\toprule
$R_c$ (m) & \shortstack{Init.\\degree} & \shortstack{Avg.\\degree} & \shortstack{Avg.\\comp.} & \shortstack{Largest\\frac.} & \shortstack{Contacts\\per run} & \shortstack{Median\\contact (s)} \\
\midrule
$75$    & $0.00$ & $0.15$ & $9.28$ & $0.14$ & $12.8$ & $12$ \\
$150$   & $0.00$ & $0.16$ & $9.22$ & $0.15$ & $12.9$ & $13$ \\
$250$   & $0.00$ & $0.17$ & $9.15$ & $0.15$ & $14.4$ & $9$ \\
$500$   & $0.05$ & $0.21$ & $8.96$ & $0.17$ & $17.6$ & $17$ \\
$1000$  & $0.23$ & $0.38$ & $8.19$ & $0.22$ & $21.8$ & $42$ \\
$2000$  & $0.89$ & $1.10$ & $5.52$ & $0.38$ & $30.2$ & $177$ \\
$4000$  & $2.91$ & $3.17$ & $1.84$ & $0.83$ & $45.8$ & $424$ \\
$8000$  & $7.58$ & $7.83$ & $1.00$ & $1.00$ & $52.0$ & $1185$ \\
$15000$ & $9.00$ & $9.00$ & $1.00$ & $1.00$ & $45.0$ & $1290$ \\
\bottomrule
\end{tabular}
\end{table}

\section*{Supplementary Table: Capability Mapping of Distributed Allocation Methods}

Table~\ref{tab:capability} maps the distributed allocation methods discussed in Section~II of the main paper onto the six capability dimensions that decide viability under persistent partitioning: the connectivity premise, the coupling between planning and physical execution, the protection of physical commitment, the dissemination of completion state, the handling of runtime task arrival, and the communication pattern. Each entry reflects the method as characterized in the source cited in Section~II of the main paper; PI-OM and CBBA-OM are the matched online wrappers of Section~VI-A of the main paper, and the experimental comparison against a static geographic-partition baseline is reported in Section~VI-B of the main paper.

\begin{table*}[!t]
\centering
\caption{Capability mapping of distributed allocation methods under persistent partitioning.}
\label{tab:capability}
\renewcommand{\arraystretch}{1.25}
\setlength{\tabcolsep}{4pt}
\footnotesize
\begin{tabular}{@{}p{2.4cm}p{3.0cm}p{2.6cm}p{1.9cm}p{2.3cm}p{1.7cm}p{2.5cm}@{}}
\toprule
Method & Connectivity premise & Planning--execution coupling & Commitment protection & Completion dissemination & Runtime task arrival & Communication pattern \\
\midrule
CBBA (Choi et al.) & Jointly connected graph union over a bounded window & Allocate-then-execute (synchronized rounds) & None & Not applicable (allocation precedes execution) & Global replan & Per-round neighbor exchange of $O(M)$ state \\
PI (Zhao et al.) & Sufficient propagation over the (possibly sparse) topology within the negotiation horizon & Allocate-then-execute & None & Not applicable & Global replan & Per-round neighbor exchange \\
PI-MaxAss (Turner et al.) & Globally consistent view for the two-stage rescheduling & Allocate-then-execute & None & Not applicable & Global replan & Two-stage per-round exchange \\
ACBBA (Johnson et al.) & Eventual network-wide message delivery & Asynchronous, locally timed replanning & None & Implicit through bid consensus & Local trigger & Asynchronous broadcasts \\
Communication-limited auctions (Otte et al.) & Message loss tolerated during the auction; allocation completes before execution & Allocate-then-execute & None & Not applicable & Global replan & Repeated auction rounds \\
Local task exchange (Liu et al.) & Opportunistic neighbor contacts & Pairwise swaps on contact; scores partition-myopic & None & Pairwise state exchange & Not addressed & Neighbor-local exchanges \\
Time-freezing replanning (Chen et al.) & Network-wide pause signal reachable & Periodic global replanning, motion frozen & None & Within each replanning round & Global replan & Round-based re-consensus \\
PI-OM / CBBA-OM (this paper) & None (opportunistic contacts) & Allocate-while-executing (event-driven) & None & Monotone completed-set gossip & Local trigger & Event-driven broadcasts \\
CC-OPI (this paper) & None (persistent partitioning tolerated) & Allocate-while-executing (event-driven) & Non-preemptive state lock & Monotone completed-set gossip & Local trigger & Event-driven, suppression-gated broadcasts \\
\bottomrule
\end{tabular}
\end{table*}

\section*{Supplementary Experiment: Hold-Out Seed Validation}

All development, mechanism exploration, and hyperparameter sensitivity work used the $30$ seeds $42$--$71$. To verify that the reported effect sizes are not an artifact of the development instances, the headline comparisons were re-run on $30$ fresh hold-out seeds ($72$--$101$) that played no role in development or tuning, over all six problem scales of the tight-constraint setting, with the same protocol and the seed-cluster pooled inference of Section~VI-A of the main paper. PI-MaxAss variants are omitted; the validated effects are the naive-transfer margin, the paradigm gain, the mechanism gain, and the full-connectivity gap.

Every effect replicates. Pooled across scales, CC-OPI leads PI-OM by $+7.1$ percentage points ($95\%$ CI $[+6.1, +8.1]$; development seeds: $+6.9$ $[+6.1, +7.8]$) and CBBA-OM by $+9.1$ $[+8.0, +10.2]$; the naive-transfer margins are $+19.6$ $[+18.2, +20.9]$ over PI-C and $+19.8$ $[+18.4, +21.2]$ over CBBA-C; the paradigm gain of PI-OM over PI-C is $+12.5$ $[+11.4, +13.6]$; and the gaps to the full-connectivity references are $-4.2$ $[-5.1, -3.2]$ to PI and $-2.4$ $[-3.3, -1.6]$ to CBBA. At the representative scale, CC-OPI leads PI-OM by $+5.9$ $[+3.8, +7.9]$. Under $50\%$ independent delivery loss at the representative scale, the hold-out runs likewise reproduce the loss study: the TCR change of CC-OPI remains statistically indistinguishable from zero ($-1.2$ $[-2.4, +0.1]$), PI-OM degrades ($-2.3$ $[-3.7, -0.9]$), and the margin of CC-OPI over PI-OM stands at $+7.0$ $[+5.0, +8.9]$.

\section*{Supplementary Experiment: Burst-Correlated Loss and Delivery Delay}

This section reports the full protocols and per-level results of the burst-loss and delivery-delay robustness tests summarized in Section~VI-F of the main paper. Both use the representative tight-constraint scale ($10$ UAVs, $30$ tasks) at $R_c=250$~m on the same $30$ seeds as the independent-loss study, comparing CC-OPI with the matched online baselines PI-OM and CBBA-OM; the frozen static baselines are omitted for the reason given in the main paper. Paired differences are computed against the unperturbed runs of the same seeds.

\emph{Burst-correlated loss.} The independent per-delivery drop is replaced by a two-state Gilbert--Elliott gate per undirected link: a link in the Bad state drops every delivery of that frame, the chain parameters are set so that the stationary Bad probability equals the mean loss level $p\in\{0.1,0.3,0.5\}$ while the mean Bad dwell is five frames ($5$~s), and the per-frame state transitions consume random draws in a fixed link order, independent of the message pattern, so runs are reproducible. Each mean loss level is thus directly comparable with its independent-loss counterpart, but outages arrive in time-correlated bursts that can silence a link for an entire contact. The loss-free level bypasses the gate and is bit-identical to the loss-free rows of the independent-loss study.

\emph{Delivery delay and reordering.} Each surviving delivery is assigned an integer delay drawn uniformly from $\{0,\dots,D_{max}\}$ seconds, with $D_{max}\in\{2,5,10\}$; a delayed message is buffered and handed to its receiver at the arrival frame regardless of connectivity at that time (the transmission already occurred), so a later broadcast with a shorter draw overtakes an earlier one and messages arrive out of order. This directly exercises the mechanism claims under reordering: the timestamped consensus rules, the monotone completed-set merge, and the position-based state lock must remain correct when state arrives late and out of sequence. The $D_{max}=0$ level consumes no draws and is bit-identical to the loss-free rows.

Table~\ref{tab:linkstress} reports the results. At every level of both perturbations, the paired TCR change of CC-OPI is statistically indistinguishable from zero, and the margin over PI-OM is preserved or widens; PI-OM degrades under bursts at every level ($-1.2$ points, $95\%$ CI $[-2.0, -0.4]$, at $p=0.5$). CBBA-OM improves slightly under mild bursts ($+0.9$ $[+0.2, +1.7]$ at $p=0.1$), the same symptom observed under independent loss: random outages interrupt the bid oscillation that commitment protection suppresses by design. The cost to CC-OPI again surfaces as redundant travel rather than lost completions---$+5.5$~km $[+1.3, +9.7]$ at $50\%$ burst loss and $+5.4$~km $[+1.3, +9.4]$ at $D_{max}=10$~s.

\begin{table*}[!t]
\centering
\caption{Burst-loss and delivery-delay results at the representative scale (means over $30$ seeds; bracketed intervals are paired bootstrap $95\%$ CIs in percentage points).}
\label{tab:linkstress}
\renewcommand{\arraystretch}{1.25}
\setlength{\tabcolsep}{3.5pt}
\footnotesize
\begin{tabular}{lccccc}
\toprule
Level & \shortstack{CC-OPI\\TCR} & \shortstack{PI-OM\\TCR} & \shortstack{CBBA-OM\\TCR} & \shortstack{CC-OPI\\$\Delta$ vs clean} & \shortstack{CC-OPI\\$-$ PI-OM} \\
\midrule
none              & $0.789$ & $0.730$ & $0.708$ & --- & $+5.9$ \\
\addlinespace[2pt]
burst $p=0.1$     & $0.785$ & $0.724$ & $0.717$ & $-0.3$ $[-0.8,+0.1]$ & $+6.1$ $[+4.2,+7.8]$ \\
burst $p=0.3$     & $0.788$ & $0.720$ & $0.710$ & $-0.1$ $[-1.2,+1.1]$ & $+6.7$ $[+4.5,+8.9]$ \\
burst $p=0.5$     & $0.786$ & $0.718$ & $0.704$ & $-0.2$ $[-1.3,+0.8]$ & $+6.8$ $[+4.9,+8.6]$ \\
\addlinespace[2pt]
delay $D_{max}=2$~s  & $0.791$ & $0.722$ & $0.712$ & $+0.2$ $[-0.6,+1.1]$ & $+6.9$ $[+4.8,+8.8]$ \\
delay $D_{max}=5$~s  & $0.782$ & $0.723$ & $0.711$ & $-0.7$ $[-1.7,+0.4]$ & $+5.9$ $[+3.8,+7.9]$ \\
delay $D_{max}=10$~s & $0.787$ & $0.722$ & $0.711$ & $-0.1$ $[-1.2,+1.0]$ & $+6.6$ $[+4.4,+8.6]$ \\
\bottomrule
\end{tabular}
\end{table*}

\section*{Supplementary Experiment: Scaling}

This section reports the scaling study referenced in Section~VI-A of the main paper. CC-OPI is run at $32\times96$, $64\times192$, and $128\times384$ (UAVs $\times$ tasks), alongside the representative $10\times30$ for reference, under the tight-constraint setting over the same $30$ seeds---the largest tier being eight times the largest configuration of the main study in both dimensions. Every run executes in a fresh process, four runs in parallel on the eight-core reference machine (AMD Ryzen 7 5800H, 16~GB RAM), so that the timing figures are not inflated by contention; the per-wake-up latency is measured by wrapping the solver invocation with a monotonic clock and isolates the allocation cost, whereas the mission wall-clock additionally includes the simulator's physics and topology maintenance. Table~\ref{tab:scaling} reports the results.

\begin{table*}[!t]
\centering
\caption{Scaling of CC-OPI under the tight-constraint setting (means over $30$ seeds).}
\label{tab:scaling}
\renewcommand{\arraystretch}{1.25}
\setlength{\tabcolsep}{3.5pt}
\footnotesize
\begin{tabular}{lcccccc}
\toprule
Scale & TCR & \shortstack{Mission\\wall-clock (s)} & \shortstack{Wake-up\\latency (ms)} & \shortstack{Peak\\RSS (MB)} & \shortstack{Bcast.\\per task} & \shortstack{Payload\\(B/bcast.)} \\
\midrule
$10\times30$   & $0.789$ & $0.5$  & $0.18$ & $128$ & $54.7$ & $456$ \\
$32\times96$   & $0.820$ & $3.3$  & $0.58$ & $128$ & $64.7$ & $1440$ \\
$64\times192$  & $0.828$ & $12.5$ & $1.05$ & $129$ & $74.3$ & $2878$ \\
$128\times384$ & $0.845$ & $60.4$ & $1.89$ & $137$ & $99.1$ & $5760$ \\
\bottomrule
\end{tabular}
\end{table*}

The systems-level costs grow as the analysis of Section~V of the main paper predicts. The mean per-wake-up solver latency grows from $0.18$ to $1.89$~ms---a factor of ten for a $12.8$-fold increase in both $N$ and $M$, consistent with the $O(MC^3 + NM)$ per-wake-up bound---and remains below $2$~ms even at $128$ UAVs, so the allocation itself is far from being a real-time bottleneck. Peak resident memory is essentially flat ($128$ to $137$~MB, dominated by the interpreter baseline), reflecting the $O(M+N)$ per-UAV state. The mean payload per broadcast tracks the analytic $16M+8N+8$-byte bound ($7{,}176$~B at the largest tier; the observed mean of $5{,}760$~B and maximum of $5{,}805$~B stay below it because the monotone completed set never fills). Two quantities grow faster than linearly in the mission size: the broadcasts per completed task rise from $54.7$ to $99.1$, because a denser swarm on the fixed field produces more topology events and more receivers per broadcast, and the mission wall-clock reaches about one minute at the largest tier---the aggregate of roughly $28{,}000$ solver invocations across the swarm (about $53$~s), which a real deployment would execute in parallel on $128$ vehicles, the remainder being the simulator's physics and $O(N^2)$ topology maintenance. The completion quality does not deteriorate: the TCR rises from $0.789$ to $0.845$, since a denser deployment on the fixed field yields richer contact opportunities, and the redundant travel per UAV saturates at about $6$~km. We report these figures as simulation-scale evidence; they do not establish deployment-scale feasibility of the radio layer, which is outside the model.